\documentclass[12pt]{article}

\usepackage{amsmath, amssymb, amsthm}
\usepackage{setspace}
\usepackage{bm}
\usepackage{comment}
\usepackage{graphicx}
\usepackage[margin=1in]{geometry}
\usepackage{setspace}
\usepackage{natbib}
\usepackage{booktabs}
\usepackage{tabularx}
\usepackage{algorithm}
\usepackage{algpseudocode}
\usepackage{multirow}

\newcommand{\bbZ}{\mathbb Z}

\newcommand{\cind}{\stackrel{d}{\to}}
\newcommand{\cinp}{\stackrel{p}{\to}}
\newcommand{\Cov}{\operatorname{Cov}}

\newcommand{\DMMD}{\mathcal D_{\operatorname{MMD}(\kappa)}}
\newcommand{\Fhat}{\widehat F}
\newcommand{\fdrf}{\texttt{fDRF}}
\newcommand{\Gam}{\operatorname{Gamma}}
\newcommand{\iid}{\stackrel{\textnormal{iid}}{\sim}}
\newcommand{\E}{\mathbb E}
\newcommand{\GP}{\operatorname{GP}}
\newcommand{\Identity}{\mathrm{I}}
\newcommand{\Leaves}{\mathcal L}
\newcommand{\median}{\operatorname{median}}
\newcommand{\MMD}{\operatorname{MMD}}
\newcommand{\Normal}{\operatorname{Normal}}

\newcommand{\Reals}{\mathbb R}

\newcommand{\sG}{\mathcal G}
\newcommand{\sH}{\mathcal H}
\newcommand{\sK}{\mathcal K}

\newcommand{\sN}{\mathcal N}
\newcommand{\sP}{\mathcal P}

\newcommand{\sS}{\mathcal S}
\newcommand{\sX}{\mathcal X}
\newcommand{\sY}{\mathcal Y}
\newcommand{\test}{\texttt{test}}
\newcommand{\train}{\texttt{train}}
\newcommand{\Var}{\operatorname{Var}}

\newcommand{\Tree}{\mathcal T}
\newcommand{\Uniform}{\operatorname{Uniform}}

\newcommand{\Frechet}{Fr{\'e}chet}
\newcommand{\Holder}{H\"{o}lder}

\newtheorem{lemma}{Lemma}
\newtheorem{proposition}{Proposition}
\newtheorem{theorem}{Theorem}

\theoremstyle{definition}
\newtheorem{assumption}{Assumption}

\newtheorem{remark}{Remark}

\usepackage[citecolor = blue, colorlinks]{hyperref}

\title{Conditional Distribution Estimation for Functional Responses with Random Forests}

\author{
  Poorbita Kundu$^{*}$ and
  Antonio R. Linero$^{\dagger}$
}

\begin{document}

\maketitle

\begin{abstract}
  Many functional data analyses reduce random functions to scalar summaries or conditional mean curves. This is limiting when we wish to understand how covariates affect the distribution of entire functional responses, including their shape, timing, or variability. We study the problem of estimating conditional laws of functional outcomes and show that these objects can be estimated and evaluated in a practical nonparametric framework. To do this, we introduce \emph{functional distributional random forests}, which estimate each conditional law as a covariate-dependent distribution over sampled functions by training a random forest to minimize a kernel-based maximum mean discrepancy within the leaf nodes of the decision tree. This supports inference on arbitrary functionals of the conditional distribution while keeping predictive samples tied to realistic curves. We consider a variety of kernels defined on function spaces, including Sobolev and operator-induced kernels. We also provide conditions for consistency of our estimator and develop scoring rules for comparing it to baseline estimators. In simulations, our method recovers distributional changes that are missed by baseline methods. In an application to NHANES accelerometer data, it identifies interesting covariate-associated changes in both median activity profiles and predictive dispersion.
\end{abstract}

\footnote{$^{*}$Fred Hutchinson Cancer Center. $^{\dagger}$Department of Statistics and Data Sciences, University of Texas at Austin. \href{mailto:antonio.linero@austin.utexas.edu}{antonio.linero@austin.utexas.edu}}

\doublespacing

\section{Introduction}

Functional data are often collected because it is believed that the shape of a profile carries scientific information. For example, a daily activity profile is a structured object whose timing, variability, and extremes may themselves be of interest. Nevertheless, many analyses simplify functional responses before modeling them by reducing each curve to a small number of scalar summaries. Those that do not typically focus on coarse quantities like the conditional mean (or \Frechet\ mean, \citealp{petersen2019frechet}) rather than providing a predictive distribution for curves. These approaches are useful, but they can miss differences between populations that appear mainly in the spread of trajectories, the timing of peaks, the roughness of curves, or the probability of unusual functional patterns. 

This motivates treating the conditional law of a functional response as a primary estimand. This makes it possible to ask more direct questions about how covariates affect the distribution of random functions, but also introduces statistical and computational challenges. First, the response space is infinite-dimensional, so a conditional distribution cannot be summarized by a small number of parameters without imposing structure that may be inappropriate. Second, functional observations often have geometries that should be respected; examples include smoothness, monotonicity, or convexity. Such considerations are central in functional data analysis (FDA), particularly in sparse functional settings and analyses of random objects \citep{yao2005functional, wang2016functional, petersen2019frechet}. Third, distributional changes may occur in features that are not well captured by pointwise means or variances, such as phase variation or tail behavior, which have long been recognized as important sources of structure in FDA \citep{tang2008pairwise, marron2015functional, dai2018intrinsic}. Finally, distribution estimation is itself nontrivial, and standard prediction errors assess point estimates rather than full conditional distributions. These issues call for methods that can represent flexible distributions over functions, adapt to covariates in a nonparametric way, and be assessed using criteria appropriate for probabilistic predictions.

To address these challenges, we develop a framework for estimating and evaluating conditional distributions of functional responses. Our main contributions are as follows.
\begin{enumerate}
\item We propose the \emph{functional distributional random forest} (\fdrf) as an estimator of conditional laws of functional outcomes. We also propose a procedure for uncertainty quantification and assess the performance of confidence intervals empirically.
\item We provide theoretical and empirical evidence of the usefulness of \fdrf{s}. Theoretically, we provide conditions for consistency of the estimated conditional laws, viewed both as probability measures on $L_2([0,1])$ and $C([0,1])$. Empirically, we demonstrate strong performance on both synthetic data examples and real data from the National Health and Nutrition Examination Survey (NHANES).
\item To compare the \fdrf\ with competing methods, we also propose, as a contribution of independent interest, proper scoring rules for evaluating the quality of conditional distribution estimates for functional outcomes.
\end{enumerate}

The \fdrf\ represents each conditional law as a covariate-dependent weighted empirical distribution over observed functions, and is an extension of the distributional random forest introduced by \citep{cevid2022distributional}. As a discrete distribution over observed outcomes, the \fdrf\ is tied to realistic functional observations while still adapting flexibly to changes in covariate space. Splits are chosen to maximize an estimate of a maximum mean discrepancy between the distributions in the child nodes, with the discrepancy tuned to emphasize distributional differences in features such as shape, roughness, or other operator-induced summaries of the response. This approach makes conditional-law estimation for functional data practical, supports inference on summaries of the estimated law, and allows detection of distributional structure missed by methods focused on conditional means.

We support these claims through a combination of theory establishing consistency of estimates, simulations showing effective learning of conditional laws relative to baselines, and an application to accelerometer data from the National Health and Nutrition Examination Survey (NHANES). On the theoretical side, establishing consistency is not immediate from existing results because (i) convergence of distributions in maximum mean discrepancy only implies convergence in distribution for random elements taking values in a locally compact Polish space, and (ii) because we are often interested in quantities, such as the maximum of a function and point-evaluation of a function, that are not continuous functionals on $L_2([0,1])$. 

\subsection{Related Work}

Random forests were originally developed for scalar regression and classification \citep{breiman2001random}, but a growing literature has adapted tree ensembles to settings with both functional predictors and functional outcomes. For functional predictors, \citet{gregorutti2015grouped} proposed grouped variable importance for random forests after representing functional covariates in a wavelet basis, and \citet{moller2016random} constructed forests using the average of the functional over random intervals as covariates. For functional responses, \citet{rahman2019functional} developed functional random forests for predicting dose-response curves, and \citet{fu2021functional} proposed FunFor for curve responses with high-dimensional scalar predictors. These approaches are closest to ours in their use of forests with functional objects, but their primary targets are point prediction, variable ranking, or conditional mean-type curve estimation. More broadly, functional regression methods based on linear, additive, or mixed-effect structures \citep{ramsay2005functional,morris2015functional,greven2017general} and random-object methods based on conditional Fr{\'e}chet means \citep{petersen2019frechet,qiu2024random} provide flexible ways to estimate central features of functional or object outcomes, but do not estimate the full conditional law of a random function.

A separate line of work uses forests for uncertainty quantification and conditional distribution estimation. Most directly related is the distributional random forest of \citet{cevid2022distributional}, which grows forests using an MMD-based split criterion and represents the fitted conditional law as a weighted empirical distribution. We extend this target-free distributional-forest from multivariate Euclidean responses to functional responses, drawing on methods for performing two-sample kernel tests for distributions over functions \citep{gretton2012kernel,wynne2022kernel}. This distinction is important because the scientific targets in functional data often concern distributions of shape, timing, roughness, extrema, or other nonlinear functionals, rather than just pointwise means or scalar quantiles. Beyond this, quantile regression forests \citep{meinshausen2006quantile} use forest-based weights to estimate the conditional distribution of a scalar response, while random survival forests \citep{ishwaran2008random} estimate conditional distributions in the presence of censoring. Other methods include parametric distributional regression forests \citep{schlosser2019distributional}, conditional density forests \citep{pospisil2019frfcde}, Wasserstein random forests \citep{du2021wasserstein}, and generalized random forests for targets defined by local estimating equations \citep{athey2019generalized}.

Finally, there is a large literature on generative modeling for functionals. Recent neural approaches, including TimeGAN \citep{yoon2019time}, variational autoencoding neural operators \citep{seidman2023variational}, and function-space autoencoders \citep{bunker2025autoencoders}, are flexible and can learn complex laws over curves. These are promising alternatives when the goal is to sample entirely new functions, but they typically require more modeling and tuning decisions, including but not limited to the choice of architecture, latent dimension, likelihood, or objective. By contrast, the \fdrf\ is designed for conditional distribution estimation and downstream inference, and automatically preserves realistic curve geometry while also keeping functionals of the conditional law easy to compute.

\section{Functional Distributional Random Forests}

\subsection{Notation and Basic Concepts}

We let $Y_1, \ldots, Y_N$ be random elements taking values in a separable Banach space $\sY$ with norm $\|\cdot\|_{\sY}$. We assume that $\sY$ embeds continuously into $L_2([0,1])$, but for much of the theoretical development this is unnecessary. Let $\sP(\sY)$ denote the set of probability distributions on $\sY$ under the Borel $\sigma$-field induced by $\|\cdot\|_{\sY}$. We also let $X_1, \ldots, X_N$ be random vectors on $\sX \subseteq \Reals^P$ and assume the $(X_i, Y_i)$'s are iid from a joint law $F_{X,Y}$. Let $\{F_x\}$ be such that $[Y_i \mid X_i = x] \sim F_x$, i.e., $\{F_x : x \in \sX\}$ is a regular conditional distribution of $Y_i$ given $X_i$.

Next, we let $\kappa : \sY \times \sY \to \Reals$ be a positive definite kernel function. Associated with $\kappa$ is a \emph{reproducing kernel Hilbert space} (RKHS) $\sH_\kappa$ of functions on $\sY$ equipped with inner product $\langle \cdot, \cdot \rangle_{\sH_\kappa}$ and norm $\|\cdot\|_{\sH_\kappa}$. The defining property of $\sH_\kappa$ is that, for each $y \in \sY$, the function $\kappa(y,\cdot)$ belongs to $\sH_\kappa$, and the reproducing property
\begin{math}
    f(y) = \langle f, \kappa(y,\cdot) \rangle_{\sH_\kappa}
\end{math}
holds for every $f \in \sH_\kappa$. The feature map $y \mapsto \kappa(y,\cdot)$ embeds elements of $\sY$ into $\sH_\kappa$.

For any $F \in \sP(\sY)$, the \emph{kernel mean embedding} of $F$ into $\sH_\kappa$ is defined by $F \mapsto \mu_F(\cdot) = \int \kappa(y, \cdot) \ F(dy)$, provided that the integral is well-defined. Define the \emph{maximum mean discrepancy} (MMD) between $F, G \in \sP(\sY)$ as $\MMD(F, G) = \|\mu_F - \mu_G\|_{\sH_\kappa}$ where
\begin{align*}
  \|\mu_F - \mu_G\|^2_{\sH_\kappa}
  =
  \int \kappa(y, y') \ F(dy) \ F(dy')
  + \int \kappa(y, y') \ G(dy) \ G(dy')
  - 2 \int \kappa(y, y') \ F(dy) \ G(dy').
\end{align*}
The kernel $\kappa$ is called \emph{characteristic} if $\MMD(F, G) = 0$ implies that $F = G$. When $\kappa$ is characteristic, $\MMD(F,G)$ is a metric on $\sP(\sY)$.

\subsection{Random Forest Tree-Growing Algorithm}
\label{sec:random-forest}

The \fdrf\ estimates the conditional law $F_x = F_{Y \mid X}(\cdot \mid x)$ by using a random forest to construct a covariate-dependent empirical distribution over the observed functional responses. The full algorithm is given in Algorithm~\ref{alg:fdrf-fit-predict} and Algorithm~\ref{alg:fdrf-grow-tree} of the Supplementary Material, and also includes an additional ``honesty'' subsampling step (see Section~\ref{sec:theoretical-properties}).
Each split is chosen to separate the training sample into child nodes whose empirical distributions are as different as possible, as measured by the MMD induced by the kernel $\kappa$. 
The resulting forest can then be used for estimating many different summaries of the conditional law.

Consider a node $b$ of a tree $\Tree$. Let $R_b\subseteq \Reals^P$ denote the covariate cell associated with the node,
\begin{math}
  I_b=\{i:X_i\in R_b\}
\end{math}
with $n_b = |I_b|$ and $Y_b = \{Y_i : i \in I_b\}$. For a candidate split along coordinate $j \in \{1,\ldots,P\}$ at cutpoint $c$, write
\begin{math}
  R_{bL}(j,c)=R_b\cap\{x:x_j\le c\},
\end{math}
and 
\begin{math}
  R_{bR}(j,c)=R_b\cap\{x:x_j>c\},
\end{math}
with associated index sets $I_{bL}(j,c)$ and $I_{bR}(j,c)$, sizes
$n_{bL}(j,c)$ and $n_{bR}(j,c)$, and response samples
$Y_{bL}(j,c)$ and $Y_{bR}(j,c)$. 
Among the admissible splits $(j,c)\in\mathcal A_b$, where $\mathcal A_b$ incorporates the usual random-forest restrictions such as randomly selected coordinates and minimum child-node sizes, the split is chosen according to
\begin{align}
  \label{eq:fdrf-split}
  (j_b,c_b)
  \in
  \arg\max_{(j,c)\in\mathcal A_b}
  \frac{n_{bL}(j,c)n_{bR}(j,c)}{n_b^2}
  \,
  \DMMD\{Y_{bL}(j,c),Y_{bR}(j,c)\}.
\end{align}
Here, for two finite samples of functions $U = \{u_1, \ldots, u_{|U|}\}$ and $V = \{v_1, \ldots, v_{|V|}\}$, the empirical MMD appearing in \eqref{eq:fdrf-split} is
\begin{align}
  \label{eq:fdrf-empirical-mmd}
  \DMMD(U,V)
  &=
  \frac{1}{|U|^2}\sum_{r,s=1}^{|U|}
    \kappa(u_r,u_s)
  +
  \frac{1}{|V|^2}\sum_{r,s=1}^{|V|}
    \kappa(v_r,v_s)
  -
  \frac{2}{|U||V|}
  \sum_{r=1}^{|U|}\sum_{s=1}^{|V|}
    \kappa(u_r,v_s).
\end{align}
Thus the split in \eqref{eq:fdrf-split} favors a partition for which the conditional distribution of $Y$ maximally changes across the two children. The factor $n_{bL}n_{bR}/n_b^2$ discourages highly unbalanced splits and is the analogue of the scaling that appears in the usual CART variance-reduction criterion. The criterion in \eqref{eq:fdrf-split} can be viewed as a version of CART in which the response $Y_i$ is replaced by its RKHS feature representation $\varphi(Y_i)$. For a general functional kernel, the split can respond to changes in the entire distribution of the curve, including changes in variability, roughness, timing, or other features encoded by $\kappa$, rather than only changes in a pointwise or integrated mean. 

Exact computation of \eqref{eq:fdrf-empirical-mmd} is expensive. Following \citet{cevid2022distributional}, we will use a random-feature representation of the MMD splitting criterion, tailored to be appropriate for functional outcomes. For simplicity, assume that the kernel is given by $\kappa(y, y') = \exp\left\{ -\gamma\|y - y'\|_{L_2}^2 / 2 \right\}$. Let $W(y)$ be the Itô integral $\int_0^1 y(t) \ dB(t)$ where $B(t)$ is standard Brownian motion on $[0,1]$ so that $W(\cdot)$ is an isonormal Gaussian process on $L_2([0,1])$, i.e.,  $W(y) \sim \Normal(0, \|y\|^2_{L_2})$ and $\Cov\{W(y), W(y')\} = \langle y, y' \rangle_{L_2}$ (\citealp{nualart2009malliavin}, Chapter 1). Then the kernel can be expressed as
\begin{align}
  \label{eq:fdrf-random-feature-representation}
  \kappa(y,y') = \E\left\{ \varphi_{W,b}(y) \, \varphi_{W,b}(y') \right\}
  \quad \text{where} \quad
  \varphi_{W,b}(y) = \sqrt{2} \cos\{\sqrt{\gamma} W(y) + b\},
\end{align}
and $b \sim \Uniform(0,2\pi)$. Consequently, $\kappa(y,y')$ can be approximated using random features $\varphi_{W_j,b_j}(y)$
with $W_j \iid \GP(0, \langle \cdot, \cdot \rangle_{L_2})$ and $b_j \sim \Uniform(0, 2\pi)$. To implement this in practice, we add another layer of random featurization and approximate
\begin{align}
  \label{eq:random-features}
  \varphi_{W_j, b_j}(y) \approx \sqrt 2 \cos\left\{ \sqrt{\frac{\gamma}{R_j}} \sum_{r = 1}^{R_j} \xi_{rj} y(t_r) + b_j\right\},
  \qquad \xi_{r j} \iid \Normal(0,1 ),
\end{align}
for $j = 1,\ldots,B$ with $B$ denoting the number of random features we wish to use. Letting $\varphi(y) = (\varphi_{W_j, b_j}(y) : j = 1,\ldots,B)^{\top}$, the MMD criterion can then be approximated as
\begin{align*}
  \DMMD(U,V)
  \approx 
  \frac{1}{B}
  \left\|
  \frac{1}{|U|}\sum_{u \in U} \varphi(u) - \frac{1}{|V|} \sum_{v \in V} \varphi(v)
  \right\|^2,
\end{align*}
which is equivalent to the usual criterion for constructing a random forest to predict a multivariate mean. In our illustrations, we take $R_j \equiv R$ and use a large dense grid of function evaluations with the $t_r$'s evenly spaced, but we could instead take $t_r \iid \Uniform(0,1)$ and $R_j$ as large as possible. The $(\xi_{rj}, b_j, R_j, t_r)$'s are sampled independently for each split. 

While we assumed a Gaussian kernel, the argument above applies more generally. The role of $\kappa$ is to define the distributional differences in the functional response that the splitting rule is sensitive to, while the role of \eqref{eq:fdrf-random-feature-representation} is to provide random test functions whose averaged squared empirical differences reproduce the desired MMD. Consequently, this construction can be applied to functional kernels with tractable random-feature representations, such as transform-based Gaussian kernels that we discuss in Section~\ref{sec:kernel-options}. If the Gaussian kernel is not desired, we can obtain other kernels by varying the distribution of $\xi_{rj}$; in the Supplementary Material, we show how to obtain rational-quadratic, Laplacian, and Cauchy kernels.

\subsection{Estimating Conditional Distributions}
\label{sec:estimating-conditional-distributions}

We now use a fitted tree ensemble $\Tree_1, \ldots, \Tree_M$ to form estimates of the conditional distribution. Let $\ell_m(x)\in\Leaves(\Tree_m)$ be the leaf of tree $m$ containing a point $x$. If $S_m$ denotes the training subsample used for tree $m$, define
\begin{math}
  I_m(x)
  =
  \{i\in S_m:X_i\in R_{\ell_m(x)}\},
\end{math}
and $n_m(x)=|I_m(x)|$. We then let the weight assigned by tree $m$ to observation $i$ be
\begin{math}
  w_{m,i}(x)
  =
  \frac{1\{i\in I_m(x)\}}{n_m(x)}. 
\end{math}
The forest weights and conditional distribution estimate are obtained by averaging over trees,
\begin{align}
  \label{eq:fdrf-weights}
  w_i(x)
  =
  \frac{1}{M}\sum_{m=1}^M w_{m,i}(x),
  \quad \text{and} \quad
  \Fhat_x(\cdot)
  =
  \sum_{i=1}^N w_i(x)\delta_{Y_i}(\cdot).
\end{align}
The $w_i(x)$'s can be interpreted as adaptive nearest-neighbor weights: $Y_i$ receives large weight at $x$ when $X_i$ often falls in the same terminal node as $x$ across the forest. This $\Fhat_x(\cdot)$ recovers the natural estimate of the kernel mean embedding from the random forest as
\begin{math}
  \widehat\mu(x)
  =
  \int \varphi(y)\,\Fhat_x(dy)
  =
  \sum_{i=1}^N w_i(x)\varphi(Y_i).
\end{math}
Once the weights have been computed, a functional $\tau(F_x)$ of the conditional law can be estimated by the plug-in estimate
\begin{math}
  \widehat\tau(x)
  =
  \tau(\Fhat_x).
\end{math}
For example, we can estimate the predictive distribution of $\max_t Y(t)$ at $X = x$ as
\begin{math}
  \sum_{i = 1}^N w_i(x) \delta_{\max_t Y_i(t)}.
\end{math}
The same fitted forest can be used to estimate conditional means, quantiles, probabilities of events involving the whole function, or other downstream functionals. Because the estimator in \eqref{eq:fdrf-weights} is supported on the observed functions $Y_1,\ldots,Y_N$, predictive draws from the fitted conditional law can be obtained by resampling $Y_i$'s with probabilities $w_1(x),\ldots,w_N(x)$.

\subsection{Kernel Options}
\label{sec:kernel-options}

The choice of kernel in the \fdrf\ defines what distributional differences the random forest algorithm will be sensitive to. For functional responses, we might want the forest to be sensitive to both location and smoothness, which is not the case for the $L_2$-Gaussian kernel. Indeed, \citet{wynne2022kernel} show that the sensitivity for kernel-based two-sample tests for functional outcomes can vary substantially depending on the choice of kernel. We see this in Section~\ref{sec:kernel-options}, where it is shown that the $L_2$-kernel is outperformed when functions differ by smoothness.

The \fdrf\ can accommodate a squared-exponential transform kernel \citep{wynne2022kernel}
\begin{align}
  \label{eq:wynne}
  \kappa(y, y') = \exp\left\{ -\frac{\|T(y) - T(y')\|^2_{\sG}}{2\ell^2} \right\},
  \qquad \text{where} \qquad
  T: \sY \to \sG
\end{align}
and $\sG$ is a real separable Hilbert space. This kernel is also characteristic if $T$ is Borel measurable, continuous, and one-to-one (\citealp{wynne2022kernel}, Theorem~9).

Many kernels of the form \eqref{eq:wynne} are possible, and Table~\ref{tab:kernels} in the Supplementary Material gives some possibilities and their use cases. We will use the transformation $T(y) = (y, Dy)$ where $Dy$ is the weak derivative of $y$. This choice is equivalent to taking $\sY$ to be contained in the Sobolev space
\begin{math}
  H^1([0,1])
  =
  \left\{
    y \in L_2([0,1]) :
    Dy \in L_2([0,1])
  \right\},
\end{math}
with $\|T(y) - T(y')\|^2_{\sG}$ a weighted squared Sobolev norm $\lambda \, a \, \|y - y'\|_2^2 + (1 - \lambda) \, b \, \|Dy - Dy'\|_2^2$. In our illustrations, we take $\lambda = 1/2$ and tune $(a,b)$ so that the overall scale of $a \|Y_i - Y_{i'}\|^2$ and $b \|DY_i - DY_{i}'\|^2$ are the same. The \fdrf\ tree growing algorithm accommodates this norm, and generally can be applied to norms that can be expressed as (sums of) $L_2$-norms in $\sG$ using the random feature construction. Other options include kernels based on FPCA features or moment expansions for detecting changes in scale. We use the median distance heuristic to choose the length scale $\ell = \median\{\|T(Y_i) - T(Y_j)\|_{\sG} : i < j\}$.

\subsection{Inference on Target Functionals}
\label{sec:inference-target-functionals}

In addition to estimating the full conditional law $F_x$, we often want uncertainty quantification for a scalar functional
\begin{math}
  \theta(x) \equiv \theta(F_x) \in \Reals.
\end{math}
In our illustrations we consider the circular mean and standard deviation of the peak time of physical activity of NHANES participants, as well as conditional quantiles and exceedance probabilities.

We use the plug-in $\widehat \theta(x) = \theta(\Fhat_x)$, so to perform inference we need an estimate of its standard error. Following \citet{naf2023confidence} (see also \citealp{athey2019generalized}), we use the ``little bags'' bootstrap. Given a forest consisting of $M$ trees, we divide the trees into $L_{\texttt{bag}}$ groups (bags) each consisting of $M_{\texttt{bag}}$ trees. For each of the bags $\ell = 1,\ldots,L_{\texttt{bag}}$, we subsample half of the data and fit an \fdrf\ on the subsample. For $M_{\texttt{bag}} = \infty$, the variance approximation
\begin{align*}
  \widehat V_{\texttt{between}}(x) = \frac{1}{L_{\texttt{bag}} - 1} \sum_{b} \left\{ \widehat \theta^{(b)}(x) - \widetilde \theta(x) \right\}^2
  \quad \text{where} \quad
  \widetilde \theta(x) = \frac{1}{L_{\texttt{bag}}} \sum_{b} \widehat \theta^{(b)}(x)
\end{align*}
is valid, where $\widehat \theta^{(b)}(x)$ is the estimate in the $b^{\text{th}}$ bag and $\widetilde \theta(x) = \frac{1}{L_{\texttt{bag}}} \sum_b \widehat \theta^{(b)}(x)$. We could then form confidence intervals of the form $\widehat\theta(x) \pm z_{1 - \alpha/2} \sqrt{\widehat V_{\texttt{between}}(x)}$ where $z_{1-\alpha/2}$ is the relevant cutoff from the standard normal distribution.

The procedure above is not exact because $M_{\texttt{bag}}$ is finite; $\widehat V_{\texttt{between}}(x)$ instead overestimates the standard error of the full forest. To account for this, one can subtract an estimate of the within-bag variability, i.e., the variability due to $M_{\texttt{bag}}$ being finite. While possible, we found empirically that this often introduced new problems that reduced the performance of interval estimates. For example, removing an estimate of the extra variability can lead to negative variance estimates, and we found this occured regularly in practice. Additionally, we found empirically that inflation of the standard errors can offset part of the finite-sample bias of random forest estimators. For that reason, our illustrations use $\widehat V_{\texttt{between}}(x)$ as our variance estimator with large values of $M_{\texttt{bag}}$ and $L_{\texttt{bag}}$. We note that an alternative approach recommended by \citet{athey2019generalized} is to use a Bayesian shrinkage estimator to ensure positivity. In Section~\ref{sec:simulation-coverage} of the Supplementary Material we perform simulation experiments to assess the adequacy of these intervals; broadly, we find good coverage in large samples, with the use of an appropriate kernel having a sometimes large impact coverage.

\section{Evaluating Distribution Estimators for Functionals}

Similar to how our construction of the random forest is tailored to conditional distribution estimation, the way in which we \emph{evaluate} estimators should be sensitive to the full conditional law $F_x$. It would be insufficient to use metrics like mean-squared $L_2$ prediction error, as it is not sensitive to features like calibration of the predictive distribution or dispersion.

To measure whether \fdrf\ gives a meaningful improvement over other approaches, we instead use \emph{proper scoring rules} \citep{gneiting2007strictly}. A function $S: \sP(\sY) \times \sY \to \Reals$ is called a proper scoring rule if
$Y_i \sim F$ implies 
\begin{math}
  \E\{S(F, Y_i)\} \le \E\{S(Q, Y_i)\}
\end{math}
for all $Q \in \sP(\sY)$.
The rule is called \emph{strictly proper} if the inequality is sharp when $Q \ne F$. For a review of proper scoring rules, see \citet{waghmare2025proper}. Given a proper scoring rule, it is natural to use the average score over the sample
\begin{align}
  \label{eq:proper-scoring-rules-eval}
  \bar S(\Fhat_x) = \frac{1}{N_{\test}} \sum_{i \in \test} S_{i}(\widehat F_x).
  \quad \text{where} \quad
  S_i(\widehat F_x) = S(\widehat F_{X_i}, Y_i),
\end{align}
\test\ denotes a subset of the data not used to construct $\Fhat_x$, and $N_{\test} = |\test|$. Given estimators $\widehat F_x$ and $\widehat F^\star_x$, we can also use the $\{S_i(\widehat F_x), S_i(\widehat F^\star_x)\}$ to explicitly test for $\E\{S_i(\widehat F_x) \mid \train\} \le \E\{S_i(\widehat F^\star_x) \mid \train\}$ using (for example) a paired $T$-test, where $\train$ denotes the training data used to build $\widehat F_x$ and $\widehat F^\star_x$.

In models where the $F_x$'s have densities $f_x$, it is natural to use the scoring rule $S(F_{X_i}, Y_i) = -\log f_{X_i}(Y_i)$, which compares models by their log-likelihood on held-out data. In our setting, however, we lack densities. As an alternative, we can use the \emph{energy score}
\begin{align}
  \label{eq:energy}
  S(F, y) = \E_{Y \sim F} \|T(Y) - T(y)\|_{\sG} - \frac{1}{2} \E_{Y, Y' \sim F} \|T(Y) - T(Y')\|_{\sG},
\end{align}
where $T(\cdot)$ is as defined in Section~\ref{sec:kernel-options}. This is a strictly proper scoring rule on $\{F : \E_{Y \sim F}\|T(Y)\|_{\sG} < \infty\}$ when $(\sG, \|\cdot\|_{\sG})$ is a real separable Hilbert space \citep{bulte2025probabilistic} and $T(\cdot)$ is one-to-one. We will use \eqref{eq:energy} with $T(y) = y$ and $T(y) = (y, Dy)$. \citet{gneiting2007strictly} also define the generic \emph{kernel score} as
\begin{math}
  S(F, y) = -\E_{Y \sim F} \kappa(Y, y) + \frac{1}{2} \E_{Y, Y' \sim F} \kappa(Y, Y').
\end{math}
When $\kappa$ is a characteristic kernel, this scoring rule is strictly proper. As with the choice of kernel in fitting the \fdrf, ideally the scoring rule should reflect the features that we want to be sensitive to; we investigate this empirically in Section~\ref{sec:simulation-study}.

Both \eqref{eq:energy} and the kernel score can be approximated by Monte Carlo integration. For example, for the energy score, for a given $(x,y)$ we can approximate
\begin{align}
  \label{eq:monte}
  S(\Fhat_x, y) \approx \frac{1}{B} \sum_{b = 1}^B \|y - \widetilde Y_b\|_{\sY}
  - \frac{1}{B(B-1)} \sum_{b < b'} \|\widetilde Y_{b} - \widetilde Y_{b'}\|_{\sY}.
\end{align}
For any $B > 1$, this is an unbiased approximation with error of order $O_P(B^{-1/2})$ for each test point; consequently, the overall error is of order $(B N_{\test})^{-1/2}$. In our illustrations, we take $B = 10$. For the \fdrf, one can also avoid Monte Carlo altogether because $\widehat F_x = \sum_{i=1}^N w_i(x)\delta_{Y_i}$ is discrete. In that case we have
\begin{math}
  \E_{Y \sim \widehat F_x} h(Y) = \sum_{i=1}^N w_i(x) h(Y_i)
\end{math}
and
\begin{math}
  \E_{Y,Y' \sim \widehat F_x} h(Y,Y') = \sum_{i=1}^N \sum_{j=1}^N w_i(x) w_j(x) h(Y_i,Y_j).
\end{math}
This avoids simulation error, but the naive computation time per test point is $O(N^2)$, which is expensive.

\section{Simulation Study}
\label{sec:simulation-study}

\subsection{Comparison of Methods By Scoring Rules}

In this section, we benchmark the \fdrf\ against alternative methods. Our goal is to determine empirically when the \fdrf\ provides improvements over alternatives, as measured by proper scoring rules.

\paragraph{Data Generating Mechanisms}

Across all of the simulations, for each observation we draw $X_i \sim \Normal(0, \Identity_5)$ and define $\eta_i = X_i^\top\beta$ where $\beta \propto (1, 0.8, -0.6, 0.4, -0.2)$ normalized so that $\|\beta\| = 1$. The functional responses are observed on an equally spaced grid for $t \in [0,1]$ with 100 grid points. We consider five different data generating mechanisms. In each case, $Y_i$ is conditionally a Gaussian process with covariance of the form $\Sigma(t, t')= \sigma^2 e^{-(t - t')^2 / (2\ell^2)}$.

\paragraph{Heteroskedastic Gaussian Process}
We sample $Y_i \sim \GP(\mu, \Sigma)$ where $\mu(t) = \sin(2\pi \, t)$, $\sigma^2 = (0.3 + e^{\eta_i/2})^2$, and $\ell = 0.2$. This will allow us to assess whether the \fdrf\ can capture changes in scale across the predictors.

\paragraph{Mean-shift Gaussian Process}
We sample $Y_i \sim \GP(\mu, \Sigma)$ where $\mu(t) = \{1+\exp(-1.5\eta_i)\}^{-1}
\times 3\exp\left\{-\frac{(t-0.5)^2}{0.04}\right\}$, $\ell = 0.15$, and $\sigma^2 = 0.4$. For this DGP, $\eta_i$ changes the height of a localized bump centered at $t = 0.5$.
 
\paragraph{Mixture GP}
Let $p(\eta_i) = (1 + e^{-2\eta_i})^{-1}$. With probability $p(\eta_i)$, we sample $Y_i$ from a Gaussian process with mean $\sin(2\pi t)$, $\ell = 0.3$, and $\sigma^2 = 1$; otherwise, we sample $Y_i$ from a Gaussian process with mean $\cos(4\pi t)$, $\ell = 0.07$, and $\sigma^2 = 1$. This DGP varies both the roughness and the overall mean of the process.

\paragraph{Warped Gaussian Process}
  
Let the baseline signal be $g(s) = e^{-(s - 0.5)^2 / 0.04}$. We draw $\xi_i \sim \Normal(0, 1)$ and define the amplitude as $a(\eta_i) = 0.22 (1 + e^{-\eta_i})^{-1}$. The warped time argument is 
\begin{math}
      s_i(t)
      =
      \min\left\{
        \max\left[
          t+a(\eta_i) \xi_i\sin(2\pi t),
          0
        \right],
        1
      \right\}.
\end{math}
We then set $Y_i(t) = g\{s_i(t)\} + \varepsilon_i(t)$ where $\varepsilon_i(\cdot)$ is a mean-zero Gaussian process with $\ell = 0.1$ and $\sigma^2 = 0.01$. This DGP makes $\eta_i$ control the magnitude of random phase variation.

\paragraph{Roughness DGP}
The response is generated as a mean-zero Gaussian process with variance
$\sigma^2=1$ and index-dependent length scale
\begin{math}
  \ell(\eta_i)
  =
  0.35(1+e^{\eta_i})^{-1}+0.03.
\end{math}
Since the length scale controls smoothness, this mechanism changes the
roughness of the functional response as a function of $\eta_i$ while keeping
the mean fixed at zero.
    
\vspace{1em}

\noindent Figure~\ref{fig:dgm-faceted} displays samples from these mechanisms, with a median curve and predictive bands. Estimates produced by the \fdrf\ on these datasets are given in the Supplementary Material, which show that it can capture the variability in the predictive distribution across all settings.

\begin{figure}[p]
  \centering
  \hspace*{-4em}
  \includegraphics[height=.89\textheight]{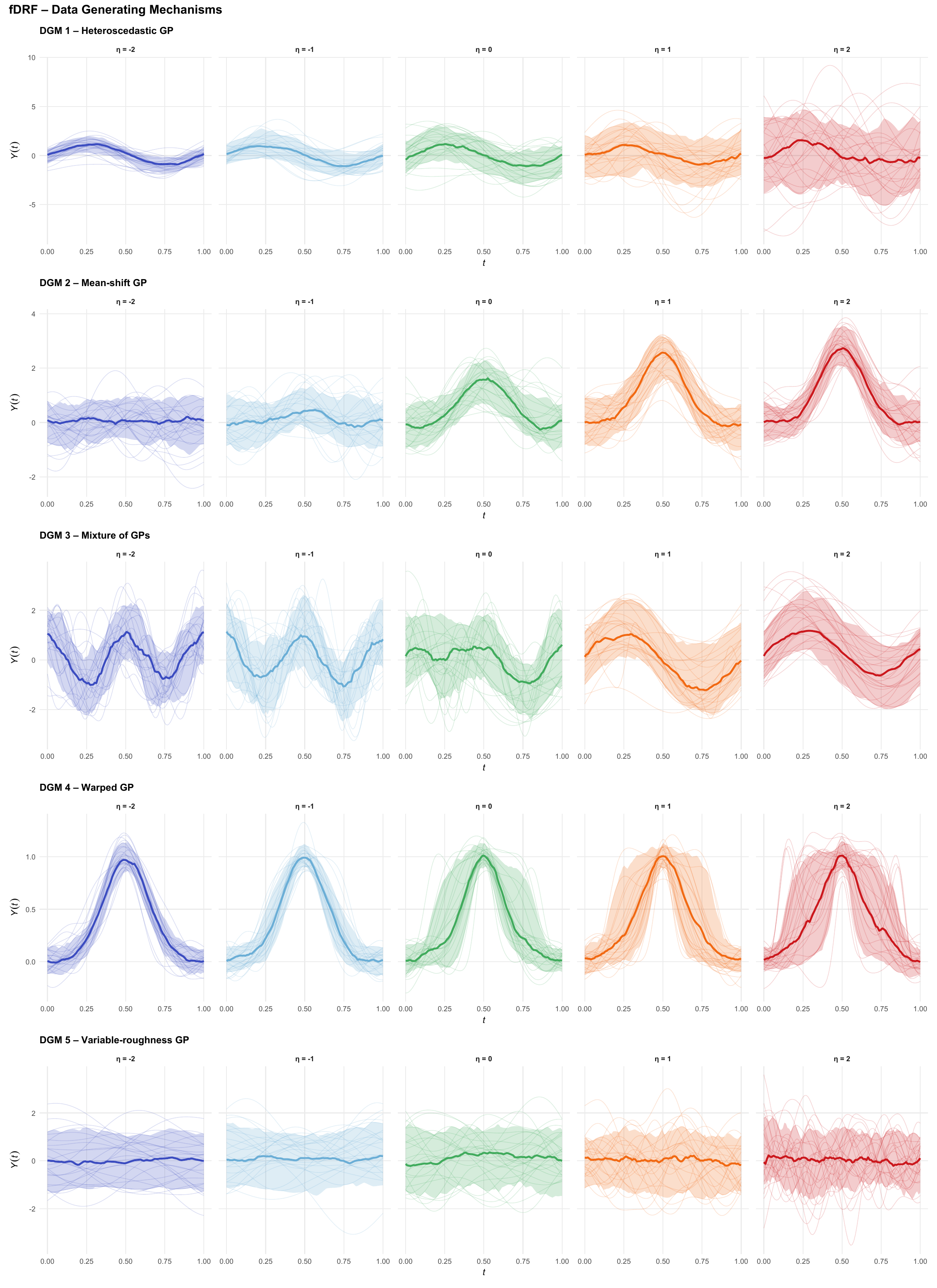}
  \caption{Data-generating mechanisms used in the simulation study. Each row shows one conditional law for $Y(t)$ across the index values $\eta = x^\top \beta \in \{-2,-1,0,1,2\}$. Thin curves are individual draws, shaded bands are pointwise 10th--90th percentiles, and thick curves are pointwise medians. }
  \label{fig:dgm-faceted}
\end{figure}

\begin{figure}[p]
  \centering
  \includegraphics[width=\textwidth]{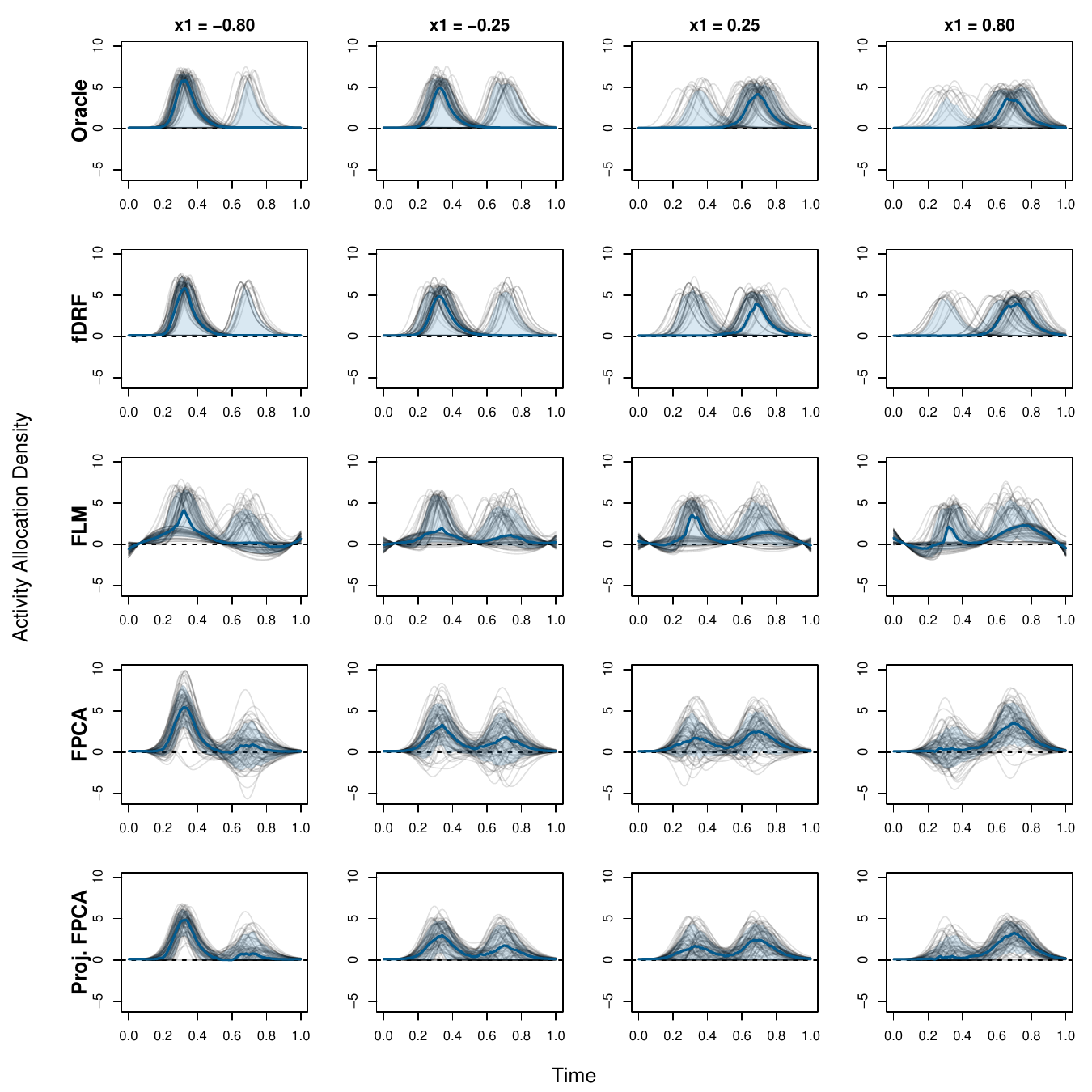}
  \caption{Samples from the conditional predictive distributions of density-valued functional data as described in Section~\ref{sec:manifold-faithfulness-experiment}; the median of the predictive distribution is given by the dark/solid line.\label{fig:manifold-density-simulation}}
\end{figure}

We compare the following methods.

\paragraph{fDRF}
We use the Gaussian kernel $\kappa(y,y') = \exp\{-\|y - y'\|^2_{\sY} / (2\omega^2)\}$, with $\omega$ chosen using the median distance heuristic. We consider
$\|y\|^2_{\sY} = \lambda \, \|y\|^2_2 + (1 - \lambda) \, \|Dy\|^2_2$, where $(Dy)(t) = \frac{\partial}{\partial t} y(t)$ with $\lambda = 1/2$ so that $y$ and $Dy$ contribute equally to the norm.
We do not tune hyperparameters and use $M = 400$ trees in all of our examples. We also fit an \fdrf\ using the functional principal component (FPCA) kernel described in Supplementary Material Table~\ref{tab:kernels}, with the number of components selected to account for 95\% of the variability in the functions.

\paragraph{$K$-Nearest Neighbors}
To check whether the data-adaptive nature of the random forest is useful, we consider a $K$-nearest neighbors (KNN) approach. We take $\Fhat_x = \sum_i w_i(x) \, \delta_{Y_i}$ where $w_i(x) = 1(X_i \in \sN_x) / \sum_{i'} 1(X_{i'} \in \sN_x)$ and $\sN_x$ consists of $K$ of the $X_i$'s that are closest to $x$. We choose $K$ to optimize the energy score by leave-one-out cross-validation on the training set, with $K$ chosen on a grid from $K = 5$ to $K = \lfloor 2 \sqrt N \rfloor$.

\paragraph{Functional Linear Models}
As a simple baseline, we consider a functional linear model (flm) implemented in the \texttt{pffr} function in the \texttt{refund} package. We fit the model
\begin{math}
  \E\{Y_i(t) \mid X_i = x\} = \alpha(t) + x^\top\beta(t)
\end{math}
with $\alpha(t)$ and $\beta(t)$ estimated via smoothing splines with smoothing parameter chosen via GCV. To sample new trajectories, we resample from residuals $r_i(t) = Y_i - \widehat \alpha(t) - X_i^\top\widehat\beta(t)$ with replacement so that the model for a new pair $(X, Y)$ is estimated as $[Y(\cdot) \mid X = x] \sim \frac{1}{N} \sum_i \delta_{\widehat \alpha(\cdot) + x^\top \widehat \beta(\cdot) + r_i(\cdot)}$.

\paragraph{Metrics for Comparison}

We evaluate methods using the held-out average score $\bar S(\Fhat_x)$ in \eqref{eq:proper-scoring-rules-eval}. We take $S(\cdot, \cdot)$ to be energy and Gaussian kernel scores computed using $T(y) = (y, Dy)$. Results are based on $20$ replicated datasets for each setting.

\paragraph{Results}

Results are summarized in Figure~\ref{fig:simulation}, where the top panels show pairwise comparisons under the energy score (left) and the kernel score (right); cells are shaded according to the average value of $\bar S(\Fhat_{\text{row}}) - \bar S(\Fhat_{\text{column}})$ (blue means the row method performed better, red means the column method performed better) and cells are tagged with $(H, M, L)$ if the median of the corresponding $P$-values for the paired $T$-test is less than $0.001$, $0.01$, and $0.05$ respectively. The bottom panels record the average rank of the methods, with an average rank of 1 indicating that the given method performed best across all simulated datasets. 

\begin{figure}[t]
  \centering
  \includegraphics[width=.9\textwidth]{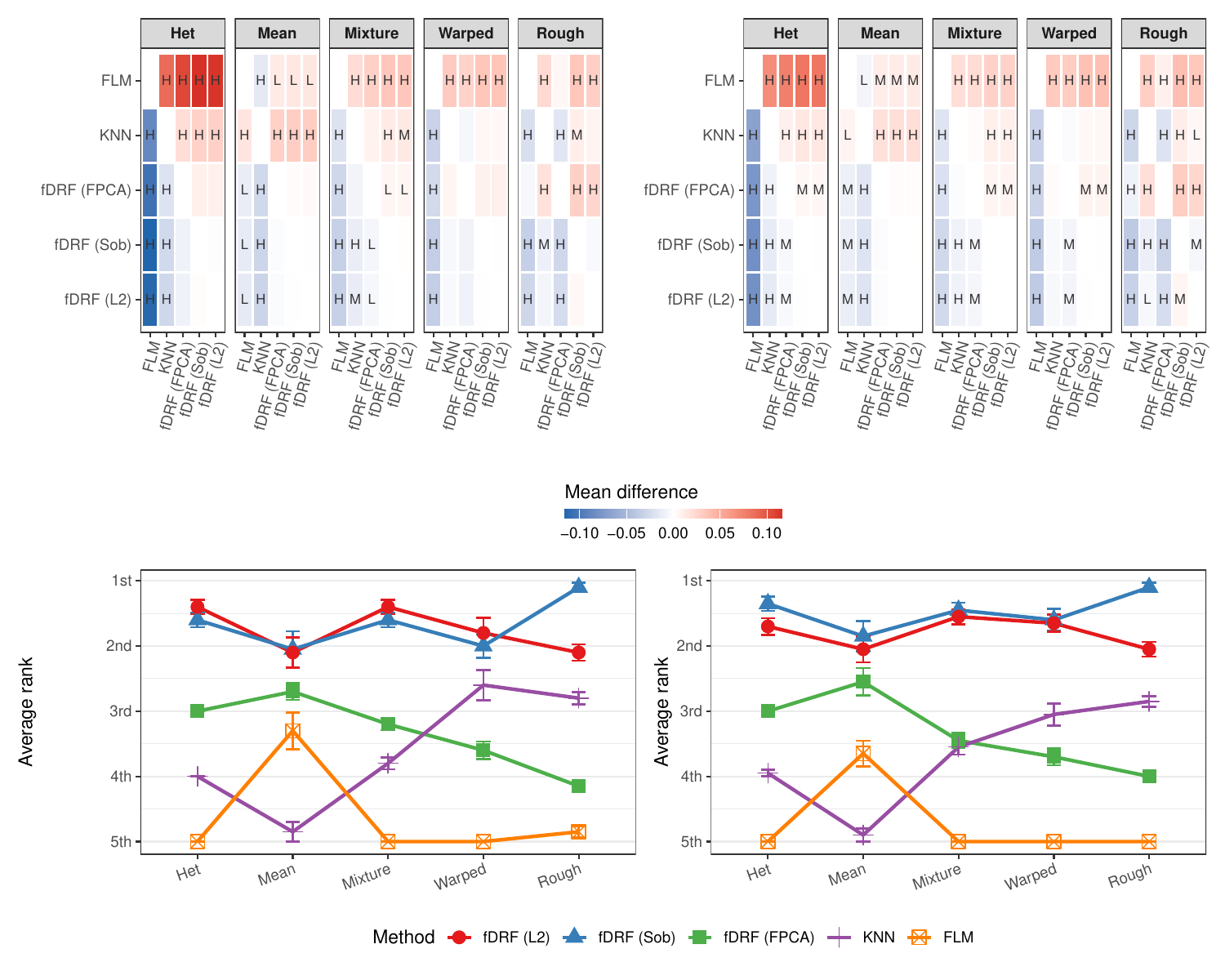}
  \caption{Top: pairwise comparison heatmap of different methods under the different simulation scenarios using the energy score (left) and kernel score (right); see text for interpretation. Bottom: average ranking of each method on each simulation setting under the energy score (left) and kernel score (right). Error bars show the mean rank $\pm$ one standard error.}
  \label{fig:simulation}
\end{figure}

We find that the \fdrf\ methods perform better than $K$-nearest neighbors and the functional linear model, that this pattern appeared repeatedly across the within-replication paired comparisons, and that the weighted Sobolev kernel tended to outperform the $L_2$-based and FPCA-based kernels. The only exception to this trend is that the \fdrf\ based on FPCA scores performed worse on the Warped GP and Roughness DGPs. Overall, the weighted Sobolev kernel performed best on all settings except the warped Gaussian process under the energy score.

When the $L_2$ norm was used instead of the Sobolev norm in constructing the scores, the results were somewhat different due to the fact that the $L_2$ norm is less sensitive to smoothness. In this case, the \fdrf\ methods perform similarly, but their improvement over the functional linear model is smaller. This is because the functional linear model can approximate the mean function well but performs very poorly in capturing the overall shape and roughness of the curves; we interpret this as evidence that the Sobolev norm is generally preferable to the $L_2$ for comparing models. Details are in the Supplementary Material.

\subsection{Manifold Faithfulness Experiment}
\label{sec:manifold-faithfulness-experiment}

An advantage of the \fdrf\ relative to competing approaches is that if $F_x$ is naturally supported on some manifold in $\sY$ then $\Fhat_x$ is as well. Hence, it is not necessary for us to bake in (say) that the $Y_i$'s are monotonically increasing, convex, or non-negative.

To illustrate this, we perform an additional experiment that compares the performance of {\fdrf}s to various functional linear models. For each subject we sample covariates $(X_{i1}, X_{i2})$ with $X_{i1} \sim \Uniform(-1, 1)$ and $X_{i2} \sim \Normal(0,1)$. For brevity, we defer the description of how the $Y_i$'s are simulated to Section~\ref{sec:appendix-dgp} of the Supplementary Material, except to note that the $Y_i$'s are densities on $[0,1]$ and that these densities depend on $X_i$ only through $X_{i1}$; sampled values of the $Y_i$'s are available in the first row of Figure~\ref{fig:manifold-density-simulation}.

For comparison, we consider the functional linear model with residual resampling to obtain a predictive distribution. In addition to this, we consider two other models:
\begin{description}
\item[FPCA Linear Model] 
  We extract functional principal components from the observed $Y_i$'s and then regress the factor scores on functions of the $X_i$'s. Conditional draws are generated by sampling factor scores at a new covariate value from a multivariate Gaussian distribution, and then reconstructing the resulting curves in the estimated principal component basis. This gives a low-rank approximation to the distribution of $Y_i$.

\item[Projected FPCA]
  To incorporate the density constraint into the fitted FPCA model, we also include a model that post-hoc corrects each sampled curve. Negative values are truncated to zero and each curve is renormalized to have a unit integral. This enforces the density constraints satisfied by the observed functional responses, while retaining the covariate-dependent model for the PC scores. 
\end{description}

We sampled $N = 800$ observations from the data generating mechanism and then fit the different predictive models to $120$ test observations with specified values of $X_{i1}$.  The estimated predictive distributions for each method (as well as the true ``oracle'' predictive distribution) are given in Figure~\ref{fig:manifold-density-simulation}, along with the median curve. Overall, we see that \fdrf\ reproduces the oracle very well. By contrast, the functional linear model and FPCA linear model perform quite poorly, and do a poor job respecting the fact that the sampled curves are densities. By comparison, Projected FPCA performs much better, but incorrectly estimates the median curve as being bimodal and is generally less accurate than the \fdrf.

\section{Theoretical Properties}
\label{sec:theoretical-properties}

We now study the theoretical properties of an ``ideal'' \fdrf\ that uses $M \to \infty$ trees and $w_i(x) = \E\{w_{m,i}(x) \mid (X_i, Y_i)_{i=1}^N\}$ (with the expectation being taken over the randomness in the random forest). Our goal is to show that the \fdrf\ leads to consistent estimation of the relevant features of $F_x$. Because the estimator returns a weighted empirical measure $\Fhat_x = \sum_{i=1}^N w_i(x) \, \delta_{Y_i}$ rather than a finite-dimensional parameter or mean curve, the relevant target is convergence of probability measures as elements of $\sP(\sY)$. This form of consistency justifies using the same forest for many downstream targets, including conditional means, quantiles of scalar summaries, and tail probabilities.

Our starting point is the consistency theory for distributional random forests established by \citet{cevid2022distributional}, which translates more-or-less without modification to the functional setting and provides convergence of probability measures with respect to MMD. After establishing this, we prove weak convergence of the laws $\Fhat_x$ to $F_x$, treating $\Fhat_x$ as either an element of $\sP(\sY)$ or $\sP(C[0,1])$ under the supremum norm; this latter type of convergence is required to perform inference on functionals such as point evaluations, suprema, peak times, and so forth. 

Assumptions~\ref{assumption:1-rf}---\ref{assumption:4-kernel} are used to establish that $\MMD(F_x, \Fhat_x) \cinp 0$. Assumption~\ref{assumption:1-rf}, which describes how the trees in the random forest are generated, is based on \citet{wager2018estimation} and is taken verbatim from \citet{cevid2022distributional}.

\begin{assumption}
  \label{assumption:1-rf}
  The tree-growing process satisfies the following conditions:
  \begin{description}
  \item[P1: Subsampling]
    The bootstrap sampling with replacement, usually used in forest-based methods, is replaced with a subsampling step, where for each tree we choose a random subset of size $s_N$ out of $N$ training points. Additionally, $s_N \asymp n^{\beta}$ for some $0 < \beta < 1$.
  \item[P2: Honesty]
    The data used for constructing each tree is split into two parts: the first is used for determining the splits and the second for constructing the weights $w_{m,i}(x)$.
  \item[P3: $\alpha$-regularity]
    Each split leaves at least a fraction $0 < \alpha \le 0.2$ of the available training sample on each side. Moreover, the trees are grown until every leaf contains between $\kappa$ and $2\kappa - 1$ observations for some fixed $\kappa \in \bbZ_+$.
  \item[P4: Symmetry]
    The (randomized) output of a tree does not depend on the ordering of the training samples.
  \item[P5: Random-split]
    At every split point, the probability that the split occurs along the feature $X_j$ is bounded below by $\pi / P$ for some $\pi > 0$ and all $j = 1,\ldots,P$.
  \end{description}
\end{assumption}

Assumptions~\ref{assumption:2-outcome}---\ref{assumption:3-covariate} concern the data generating mechanism, while Assumption~\ref{assumption:4-kernel} concerns the MMD kernel.

\begin{assumption}[Outcome Conditional Distribution]
  \label{assumption:2-outcome}
  The mapping $x \mapsto \mu(F_x) \in \sH$ is Lipschitz, i.e.,
  \begin{math}
    \|\mu(F_x) - \mu(F_{x'})\|_{\sH}
    \le L \|x - x'\|_2
  \end{math}
  for all $x, x' \in [0,1]^P$. 
\end{assumption}

\begin{assumption}[Covariate Distribution]
  \label{assumption:3-covariate}
  The covariates $X_i$ are iid random vectors with a continuous joint density $g : [0,1]^P \to \Reals$ that is bounded away from $0$.
\end{assumption}

\begin{assumption}[Kernel]
  \label{assumption:4-kernel}
  The kernel $\kappa : \sY \times \sY \to \Reals$ is bounded and continuous with respect to the topology on $\sY \times \sY$ induced by $\|\cdot\|_{\sY}$. Additionally, $\kappa$ is characteristic.
\end{assumption}

These assumptions are sufficient to port over Theorem 2 of \citet{cevid2022distributional}, without modification, to the case where $\sY$ is a separable Banach space.

\begin{theorem}[Theorem 2 of \citet{cevid2022distributional}]
  \label{thm:mmd}
  Suppose that Assumptions~\ref{assumption:1-rf}---\ref{assumption:4-kernel} hold. Then we obtain consistency with respect to the RKHS norm in the sense that, for all $x \in [0,1]^P$, $\|\mu(\widehat F_x) - \mu(F_x)\|_{\sH} = O_P(N^{-\gamma})$ for $\gamma = \frac{1}{2} \min\left( 1 - \beta, \frac{\log((1-\alpha)^{-1})}{\log(\alpha^{-1})} \times \frac{\pi}{P} \times \beta \right)$.
\end{theorem}

Theorem~\ref{thm:mmd} is not sufficient for most of our purposes for two reasons. First, convergence in MMD does not imply $\Fhat_x \to F_x$ in distribution, which is required to perform inference for most of the quantities we are interested in --- in general, $\MMD(\cdot,\cdot)$ does not metrize convergence in distribution. 
When $Y_i$ takes values in a locally compact Polish space then this correspondence \emph{does} hold under mild conditions, but unfortunately for us infinite-dimensional separable Banach spaces are not locally compact. We will instead use arguments that do not rely on establishing that $\MMD(\cdot,\cdot)$ metrizes convergence in distribution.

The second reason is that many quantities of interest need a stronger result than $\Fhat_x \to F_x$ in distribution as elements of $\sP(\sY)$. For example, if $\sY = L_2([0,1])$ then weak convergence as elements of $\sY$ will not be sufficient for estimating the predictive distribution of either the pointwise evaluations $Y_i \mapsto Y_i(t)$ or the supremum $Y_i \mapsto \sup_t Y_i(t)$, as these functionals are not continuous with respect to $\|\cdot\|_2$. Establishing convergence for these quantities requires further conditions on the smoothness of the $Y_i$'s.

We first address convergence in distribution as elements of $\sP(\sY)$ (see \citealp{billingsley2013convergence} for a thorough treatment of convergence of distributions). More precisely, we want $d(\Fhat_x, F_x) \cinp 0$ where $d(\cdot,\cdot)$ metrizes convergence in distribution; in this case we will say $\Fhat_x \to F_x$ weakly-in-probability. As an additional condition, we make the following \emph{uniform tightness} assumption.

\begin{assumption}[Uniform Tightness]
  \label{assumption:5-tightness}
  For all $x \in [0,1]^P$ and $\epsilon > 0$, there exists a neighborhood $U$ of $x$ and a compact set $K_\epsilon$ such that $\sup_{z \in U} \Pr(Y_i \notin K_\epsilon \mid X_i = z) \le \epsilon$.
\end{assumption}

\begin{theorem}
  \label{thm:weak}
  Suppose that Assumptions~\ref{assumption:1-rf}---\ref{assumption:5-tightness} hold. Then for all $x$ we have $\Fhat_x \to F_x$ weakly-in-probability as elements of $\sP(\sY)$ with respect to the weak topology induced by $\|\cdot\|_{\sY}$.
\end{theorem}

\begin{remark}
  One simple sufficient condition for checking Assumption~\ref{assumption:5-tightness} is to assume that the $Y_i$'s take values in a smoother function space $\sS$ that embeds compactly into $\sY$ and that $\sup_x \E(\|Y_i\|^p_{\mathcal S} \mid X_i = x) < \infty$ for some $p > 0$; for example, if $\sY = L_2([0,1])$ then we might take $\sS = H^s$ to be a Sobolev space with smoothness index $s > 0$ \citep{di2012hitchhikerʼs}.
\end{remark}

Assumption~\ref{assumption:6-ctightness} makes upgrading to convergence in $C([0,1], \|\cdot\|_\infty)$ straightforward. The first case of Assumption~\ref{assumption:6-ctightness} is used to handle the case where $\sY = L_2([0,1])$ while the second is used to handle the case where $\sY$ is a space of smooth functions (e.g., \Holder\ or Sobolev).

\begin{assumption}[Uniform Tightness on $C$]
  \label{assumption:6-ctightness}
  Either (i) $C$ is continuously embeddable into $\sY$, i.e., the inclusion mapping $\iota: C \to \sY$ is continuous, and Assumption~\ref{assumption:5-tightness} holds with $K_\epsilon \subseteq C([0,1])$ interpreted as a subset of $C$ relative to $\|\cdot\|_\infty$, or (ii) $\sY$ is continuously embeddable into $C$.
\end{assumption}

\begin{theorem}
  \label{thm:continuous}
  Suppose that Assumptions~\ref{assumption:1-rf}---\ref{assumption:6-ctightness} hold. Then for all $x$ we have $\Fhat_x \to F_x$ weakly-in-probability as elements of $\sP(C)$ with respect to the weak topology induced by $\|\cdot\|_\infty$.
\end{theorem}

\begin{remark}[Sufficient Conditions]
  The simplest way to verify Assumption~\ref{assumption:6-ctightness} (i) is to assume that the $Y_i$'s are supported on a sufficiently smooth class of functions and make a moment assumption, which is the same strategy as for Assumption~\ref{assumption:5-tightness}. If the $Y_i$'s are supported on a set $\sS$, we only need that $\sS$ embeds compactly into $C$ rather than $\sY$. This holds for $\sS$ a Sobolev space with index $s > 1/2$ or a \Holder\ space with smoothness index $s > 0$.
\end{remark}

\paragraph{Pointwise Inference for Selected Functionals}

Although the \fdrf\ estimates an entire conditional law, the inferential targets we consider are low-dimensional summaries. Most can be written as
\begin{math}
  \theta(x) = \Psi\{\eta_g(x)\}
\end{math}
where $\eta_g(x) = \E\{g(Y_i) \mid X_i = x\}$, $g(\cdot)$ is a summary of a given curve, and $\Psi(\cdot)$ is smooth. The corresponding estimator is
\begin{math}
  \widehat \eta_g(x) = \sum_{i = 1}^N w_i(x) \, g(Y_i)
\end{math}
and $\widehat \theta(x) = \Psi\{\widehat \eta_g(x)\}$.

The confidence interval construction in Section~\ref{sec:inference-target-functionals} can be justified by appealing to results for (generalized) random forests in \citet[][Theorem~1]{wager2018estimation} and \citet[][Theorem~5]{athey2019generalized}, which make minimal assumptions about how the trees in the ensemble are built. For example, the linearized estimator of $\Psi\{\eta_g(x)\}$ given by 
\begin{math}
  a_x^\top \{\widehat \eta_g(x) - \eta_g(x)\}
  =
  \sum_{i = 1}^N w_i(x) \{a_x^\top g(Y_i)\}
\end{math}
where $a_x = \left.\nabla_\eta \Psi(\eta)\right|_{\eta = \eta_g(x)}$ is exactly an honest regression-forest estimator with scalar pseudo-response $Z_{a_x} = a_x^\top g(Y)$. In the Supplementary Material, we argue that under Assumption~\ref{assumption:1-rf} and additional assumptions on $\nabla \Psi$ and $g(\cdot)$, a central limit theorem holds with
\begin{align*}
  \sigma_{N,\theta}(x)^{-1}\{\widehat \theta(x) - \theta(x)\}
  \cind\ \Normal(0,1),
\end{align*}
for some sequence $\sigma_{N,\theta}(x)$. We also give conditions under which the targets in this manuscript satisfy these assumptions.

\section{Analysis of NHANES}

NHANES accelerometer data provide a useful setting for illustrating how the \fdrf\ can be used to answer questions that go beyond conditional means or scalar summaries as part of an analysis of functional data. Here, each of 4,099 subjects contribute an activity profile during a monitoring period. Many natural questions are distributional in nature, such as how covariates like age, income, smoking status, and other demographic characteristics shift the typical intensity and timing of activity, the variation in the trajectories, and other derived features of the profile.

\paragraph{Overall Predictives}
We fit an \fdrf\ using 24-hour average accelerometer profiles from the NHANES 2005-2006 survey as the response, controlling for smoking status, age, race, education level, income, marital status, and sex as predictors; details on these covariates are given in the Supplementary Material. The response curves were defined by $Y_i(t) = \log\{1 + W_i(t)\}$ where $W_i(t)$ denotes the number of steps an individual took, averaged over a brief time interval, as measured by a physical activity monitor worn by the subjects; the log transform was used to control outliers. We fit the \fdrf\ with $\lambda = 0.5$ and $B = 200$ random features for each splitting rule.

Figure~\ref{fig:nhanes-age-income-distribution} shows the estimated predictive distribution within a stratum of white/male/married high-school graduates as poverty income ratio (PIR) and age vary.
There are several interesting patterns in the median activity curves; in particular, we observe an interaction between PIR and age within this stratum, with higher-income individuals having an increasing activity profile as PIR increases, while for younger individuals the trend is reversed. The overall predictive distributions in Figure~\ref{fig:nhanes-age-income-distribution} also show that the degree of \emph{uncertainty} in the profiles varies by both age and PIR; for example, there is much more variability in trajectories for high-income younger individuals than for high-income older individuals.

\begin{figure}[t]
  \centering
  \includegraphics[width=\textwidth]{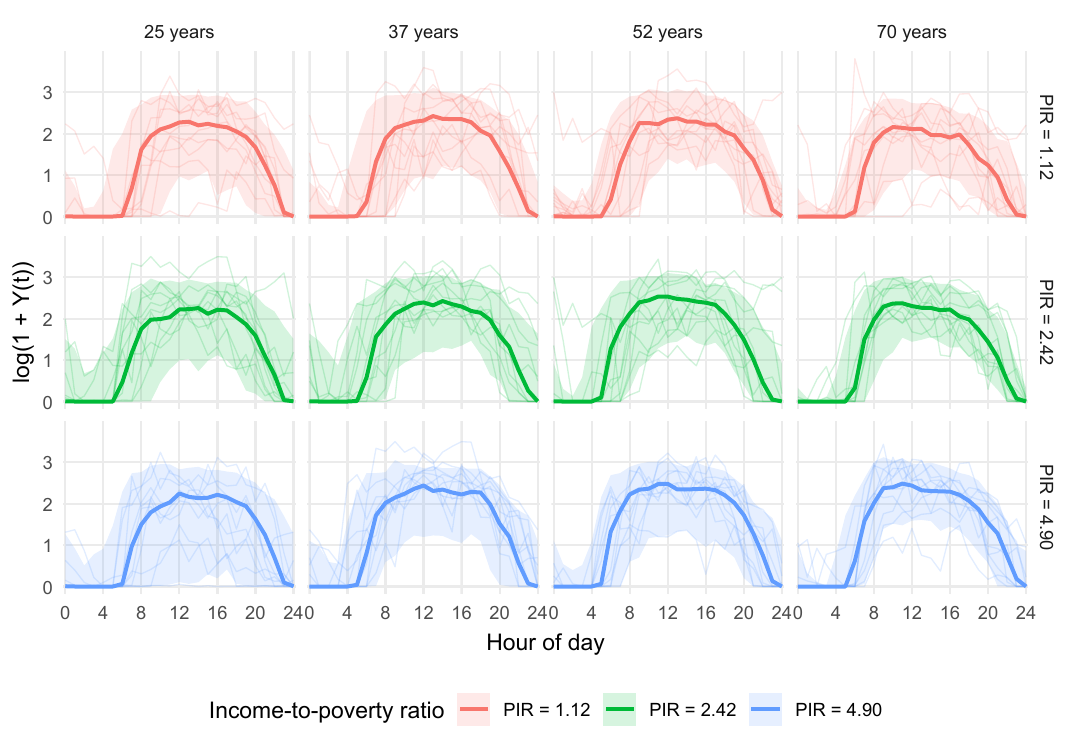}
  \caption{Predictive distributions for NHANES hourly activity profiles. Panels vary age and poverty-income ratio while holding the remaining covariates fixed at their template values. Shaded bands are pointwise 10th--90th percentiles, thick lines are pointwise medians, and thin lines show sampled predictive trajectories.}
  \label{fig:nhanes-age-income-distribution}
\end{figure}

\paragraph{Goodness of Fit}
Next, we determine whether the \fdrf\ gives better out-of-sample predictions than alternatives. We use $20$ random $80/20$ training-testing splits of the sample and fit the \fdrf, KNN, and functional linear model (FLM) residual-bootstrap estimators on the training data, and evaluate each fit on the testing data using proper scoring rules. Table~\ref{tab:funcde-flm-cauchy-comparison} reports the mean score difference, the fraction of splits in which the first method in the comparison had a lower held-out score, and the $P$-value from the resampling-corrected $t$-test for the mean score difference \citep{bouckaert2004evaluating}.
The \fdrf\ has lower held-out scores than both KNN and FLM on almost all splits, with aggregate $P$-values below $0.02$ in all four comparisons against these baselines. KNN also improves on FLM, although the evidence is weaker.

\begin{table}[t]
\centering
\small
\begin{tabular}{llrrr}
\toprule
Score                   & Comparison   & Mean difference & Prop. better & Aggregate $P$-value \\
\midrule
\multirow{3}{*}{Energy} & \fdrf\ vs. FLM & -0.013          & 1.0000       & 0.005               \\
                        & \fdrf\ vs. KNN & -0.010          & 0.9500       & 0.016               \\
                        & KNN vs. FLM  & -0.003          & 0.7000       & 0.535               \\
  \midrule
\multirow{3}{*}{Kernel} & \fdrf vs. FLM & -0.011          & 1.0000       & $<0.001$            \\
                        & fDRF vs. KNN & -0.007          & 0.9500       & 0.014               \\
                        & KNN vs. FLM  & -0.004          & 0.7500       & 0.246               \\
\bottomrule
\end{tabular}
\caption{Repeated train/test comparison of conditional distribution estimators.\label{tab:funcde-flm-cauchy-comparison}}

\end{table}

\paragraph{Day-Averaged Median Trajectories}
To check where the median estimates are separated from their sampling uncertainty, Figure~\ref{fig:nhanes-raw-trajectory-age-contrasts} shows the corresponding differences relative to the 25-year-old template within each PIR stratum. In the Supplementary Material, Figure~\ref{fig:nhanes-raw-trajectory-median} shows the conditional pointwise median together with pointwise $95\%$ confidence bands. These bands validate the difference in median profile across age in low-income individuals.

\begin{figure}[t]
  \centering
  \includegraphics[width=\textwidth]{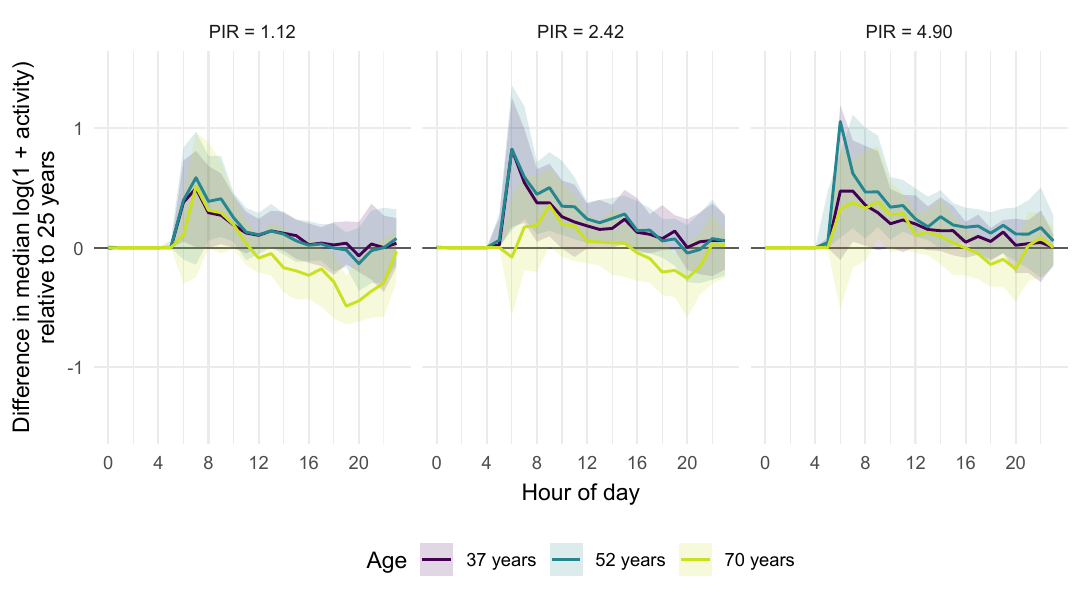}
  \caption{Pointwise differences in median day-averaged log activity trajectories relative to the 25-year-old template within each poverty-income ratio. Bands are pointwise 95\% confidence intervals and the horizontal line marks no difference.}
  \label{fig:nhanes-raw-trajectory-age-contrasts}
\end{figure}

\paragraph{Peak Time and Maximal Activity}
Let $\mathcal T\subset[0,24)$ denote the observation grid and define $P_i = \min\{t\in\mathcal T:Y_i(t)=\max_{s\in\mathcal T}Y_i(s)\}$ to be the \emph{peak time} of individual $i$, and let $M_i = \max_t Y_i(t)$ denote the \emph{maximal effort} of individual $i$. In Figure~\ref{fig:nhanes-peak-effort-predictive-violins} of the Supplementary Material we display the predictive distribution of $P_i$ and $M_i$ for the individuals in the different strata used above.
To account for the fact that time-of-day is circular (e.g., 1am and 11pm are equally close to 12am), we define $\mu_{P_i}$ and $\sigma_{P_i}$ to be the \emph{circular mean and standard deviation} of $P_i$,
\begin{align*}
  \mu_{P_i} = \left[\frac{12}{\pi}\operatorname{atan2}\left( \E\{\sin(\pi P_i / 12)\}, \E\{\cos(\pi P_i / 12)\} \right)\right]_{24}
  \qquad
  \sigma_{P_i} = \frac{12}{\pi} \sqrt{-2 \log R_{P_i}}
\end{align*}
where $R_{P_i} = \sqrt{\E\{\sin(\pi P_i / 12)\}^2 + \E\{\cos(\pi P_i / 12)\}^2}$ \citep{lee2010circular}. Here $[u]_{24}=u-24\lfloor u/24\rfloor\in[0,24)$ is $u$ modulo 24. From Figure~\ref{fig:nhanes-peak-effort-functional-intervals} we can see that older individuals tend to have earlier average peak times and that mean maximal effort increases with age up to a point before decreasing among elderly individuals, with a steeper drop-off for lower-income individuals. Generally, higher-income individuals have at every age a higher degree of variability in their peak time as well, with a sizeable dip in variability occurring in older individuals.

\begin{figure}[t]
  \centering
  \includegraphics[width=\textwidth]{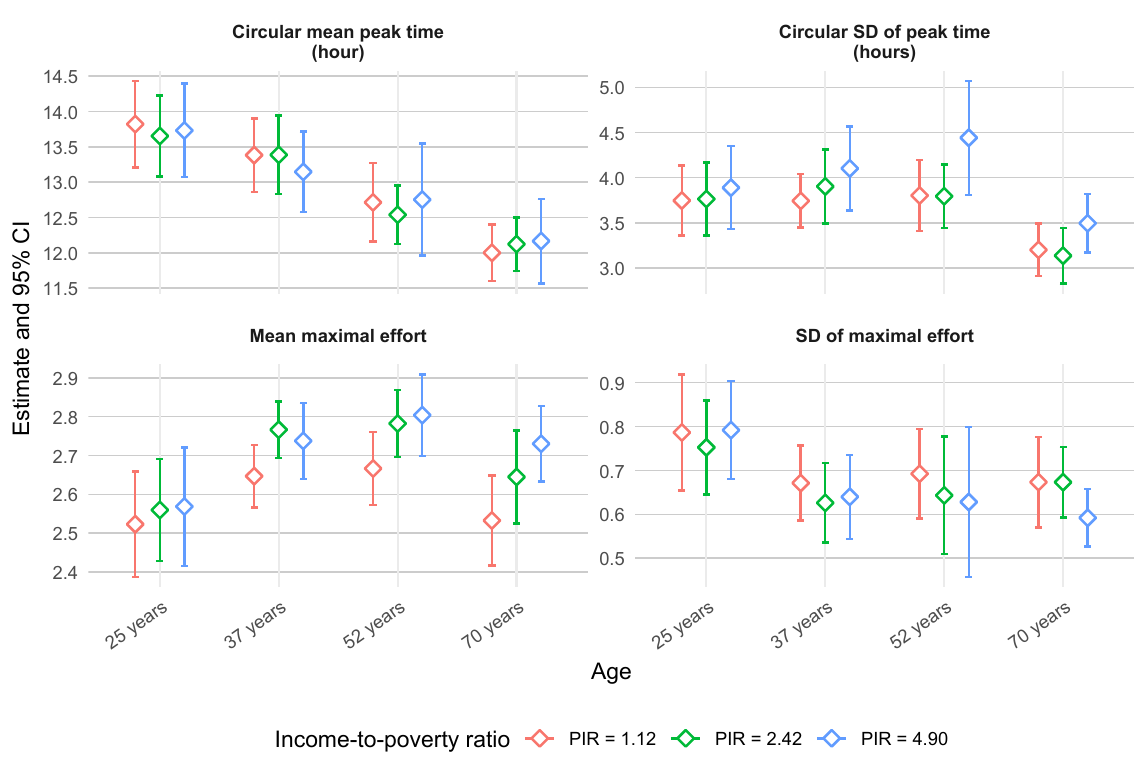}
  \caption{Inference for the circular mean and circular standard deviation of peak time and the mean and standard deviation of maximal effort.}
  \label{fig:nhanes-peak-effort-functional-intervals}
\end{figure}

\paragraph{Exceedance Function}

A useful summary of an entire activity profile is the \emph{exceedance function} \citep{kundu2025exceedance}
\begin{math}
  E_i(u) = \frac{1}{24}\int_0^{24} 1\{Y_i(t) > u\} \ dt.
\end{math}
This measures the fraction of the day individual $i$ spends with activity above threshold $u$. Unlike the peak time or maximal effort, $E_i(\cdot)$ summarizes how activity is distributed across intensities rather than where it is concentrated in time, and it does so without committing in advance to a particular cutpoint between ``sedentary,'' ``light,'' and ``moderate-to-vigorous'' activity. Both the level of $E_i(u)$ and its dependence on $u$ are of interest, as two subgroups may have similar average time above a low threshold but very different probabilities of sustaining higher-intensity activity.

Figure~\ref{fig:nhanes-exceedance-median} displays the conditional median of $E_i(u)$ with pointwise $95\%$ confidence intervals within strata. Across all PIR levels, individuals aged $37$, $52$, and $70$ tend to spend a larger fraction of the day above moderate intensity thresholds than $25$-year-olds, with the difference concentrated in the range $u \in [2, 6]$ and vanishing above $u \approx 7$ where exceedance is near zero for everyone. Additionally, the magnitude of this age gap increases with PIR. Finally, the low-PIR stratum shows a qualitatively different pattern at the upper end of the threshold range: the $70$-year curve dips meaningfully below the $25$-year curve for $u \in [5, 7]$, indicating that older low-income individuals spend less time at the higher intensities that younger low-income individuals still reach. This kind of crossover might be missed by a cutpoint analysis, but is easy to read off the conditional exceedance functions estimated by the \fdrf.

\begin{figure}[t]
  \centering
  \includegraphics[width=\textwidth]{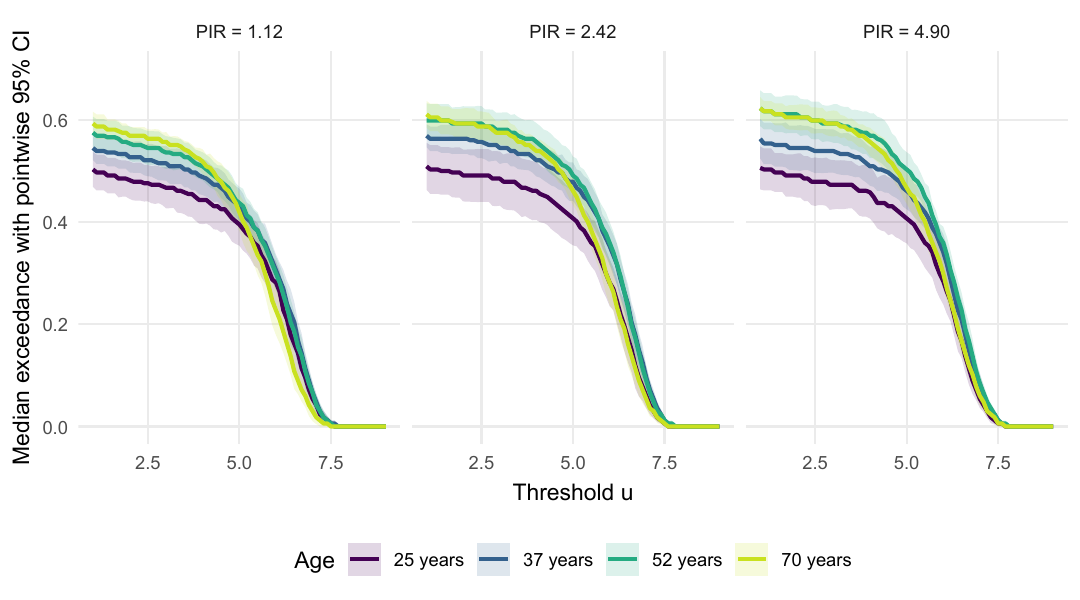}
  \caption{Median exceedance functions for weekly NHANES activity curves, with pointwise 95\% confidence bands. Facets show poverty-income ratio and color encodes age.}
  \label{fig:nhanes-exceedance-median}
\end{figure}

\begin{figure}[t]
  \centering
  \includegraphics[width=\textwidth]{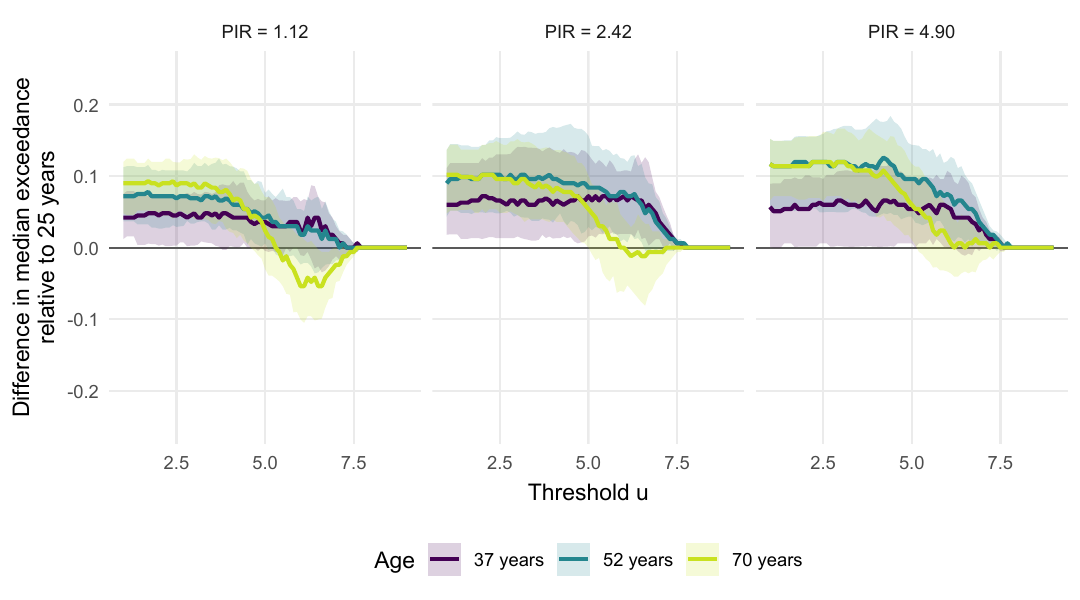}
  \caption{Pointwise differences in median exceedance functions relative to the 25-year-old template within each poverty-income ratio. Bands are pointwise 95\% confidence intervals and the horizontal line marks no difference.}
  \label{fig:nhanes-exceedance-age-contrasts}
\end{figure}

\section{Discussion}

In this work we developed the \fdrf\ for conditional distribution estimation with functional responses, developed its theoretical properties, and illustrated its strong performance on a variety of synthetic and real-data examples. Establishing theoretical results for \fdrf\ required using tools from functional analysis to account for the fact that $\MMD$ only metrizes convergence in distribution on locally compact Polish spaces.

The main strength of the \fdrf\ is that it targets the conditional distribution of a functional response rather than a prespecified scalar summary or conditional mean curve, and so allows the analyst to ask very rich questions. Relative to KNN estimators, \fdrf{s} build neighborhoods that are adaptive to the covariates. Relative to FLMs, they have very flexible mean and error structure. The MMD split criterion also lets the analyst encode relevant notions of functional similarity. Moreover, because $\Fhat_x$ is supported on observed functions, predictive draws remain realistic and automatically preserve constraints present in the data.

\paragraph{Limitations and Future Work}
A limitation of \fdrf{s} is that the estimator is discrete. While predictive samples remain realistic, \fdrf{s} cannot generate genuinely new curves outside the training support. Weight construction can also create difficulties in settings requiring common support across covariate values, such as multiply robust causal or mediation estimators based on inverse probability weighting; making the \fdrf\ usable in these settings will be the subject of future work. The method also depends on the choice of kernel, and different kernels emphasize different aspects of the functional distribution. In practice this means that kernel choice should be guided by the scientific question or examined through sensitivity analyses. Finally, the development here implicitly assumes regularly-sampled dense observations, whereas functional data are often irregularly sampled, partially observed, or subject to registration uncertainty. In this work, our theoretical results are derived for continuous-time curves, while extending these results to account for discretely-sampled approximations remains an important direction for future research.

\paragraph{Acknowledgments}
This work was partially supported by NSF grant DMS-2144933.

\bibliographystyle{apalike}
\bibliography{mybib}

\clearpage

\appendix

{\LARGE \noindent \bf Supplementary Material}

\section{Proofs}

This section provides the details behind the weak-consistency results stated in Theorems~\ref{thm:weak} and~\ref{thm:continuous}. The argument is organized as follows. Lemma~\ref{lemma:weight-localization} is standard, and shows that the forest weights concentrate on training observations whose covariates are local to the target point $x$. This localization is then combined with the local compact-containment assumptions in Lemma~\ref{lemma:tightness-1} to show that the random fitted measures place most of their mass on compact sets, with high probability. Lemma~\ref{lemma:mmd-continuity} establishes continuity of the MMD functional, and Lemma~\ref{lemma:tightness-2} lifts tightness of the random measures to tightness of their laws as random elements of $\sP(\sY)$ (or of $\sP(C[0,1])$) via Lemma~\ref{lemma:mmd-continuity}, Lemma~\ref{lemma:tightness-1}, and Prohorov's theorem.  These ingredients allow the final subsequence argument to identify every possible weak limit as the point mass at the true conditional law $F_x$.

\begin{lemma}
  \label{lemma:weight-localization}
  For decision tree $m$ in the random forest let $w_{m,i}(x) = \frac{1(i \in I_m(x))}{n_m(x)}$ where $I_m(x) = \{i \in S_m : X_i \in R_{\ell_m(x)}\}$ is the set of observations in the training set of tree $m$ that are associated to the same leaf node as $x$, and let $w_i(x) = \E\{w_{m,i}(x) \mid (X_i, Y_i)_{i=1}^N\}$. Then under Assumptions~\ref{assumption:1-rf}---\ref{assumption:4-kernel}, for any $x \in [0,1]^P$ and open neighborhood $U$ of $x$, we have $\sum_i w_i(x) 1(X_i \notin U) \cinp 0$.
\end{lemma}

\begin{proof}
  Fix $x\in[0,1]^P$ and let $U$ be an open neighborhood of $x$. Choose $r>0$ such that $B(x,r)\cap[0,1]^P\subseteq U$. For tree $m$, define
  \[
    A_{m,N}(x,U)
    =
    \sum_{i=1}^N w_{m,i}(x)1(X_i\notin U),
    \qquad
    \Delta_{m,N}(x)
    =
    \operatorname{diam}\{R_{\ell_m(x)}\}.
  \]
  Since $x\in R_{\ell_m(x)}$, if $\Delta_{m,N}(x)<r$, then every $X_i\in R_{\ell_m(x)}$ lies in $B(x,r)\cap[0,1]^P\subseteq U$. Therefore
  \[
    A_{m,N}(x,U)
    =
    \frac{1}{n_m(x)}
    \sum_{i\in I_m(x)} 1(X_i\notin U)
    \le
    1\{\Delta_{m,N}(x)\ge r\}.
  \]
  By Lemma~1 of \citet{wager2018estimation}, $\Delta_{m,N}(x) \cinp 0$, so $A_{m,N}(x,U) \cinp 0$. Now let
  \[
    A_N(x,U)
    =
    \sum_{i=1}^N w_i(x)1(X_i\notin U)
    =
    \E\{A_{m,N}(x,U) \mid (X_i, Y_i)_{i=1}^N\}.
  \]
  For any $\epsilon>0$, Markov's inequality gives
  \[
    \Pr\{A_N(x,U)>\epsilon\}
    \le
    \frac{1}{\epsilon}\E\{A_N(x,U)\}
    =
    \frac 1 \epsilon \E\{A_{m,N}(x,U)\} \to 0
  \]
  because $A_{m,N}(x,U) \cinp 0$ and $A_{m,N}(x, U)$ is bounded. 
  Thus $A_N(x,U) = \sum_{i = 1}^N w_i(x) \, 1(X_i \notin U) \cinp 0$ as claimed.
\end{proof}

\begin{lemma}
  \label{lemma:tightness-1}
  Under Assumptions~\ref{assumption:1-rf}---\ref{assumption:5-tightness}, the random probability measures $\Fhat_x$ are asymptotically tight in probability. In particular, for every $\eta > 0$ and $\delta > 0$ there exists a compact set $K \subseteq \sY$ such that
  \begin{align*}
    \sup_{N \ge 1} \Pr\{\Fhat_x(K^c) > \eta\} \le \delta.
  \end{align*}
  If Assumption~\ref{assumption:6-ctightness} is added, then this conclusion also holds with $K \subseteq C([0,1])$ chosen to be compact relative to the $\|\cdot\|_\infty$-norm on $C([0,1])$.
\end{lemma}

\begin{proof}
  We prove the result under Assumption~\ref{assumption:1-rf}---\ref{assumption:5-tightness}. The result under Assumption~\ref{assumption:6-ctightness} (ii) is a corollary of this, while the same argument applies under Assumption~\ref{assumption:6-ctightness} (i), with $K$ compact relative to $\|\cdot\|_\infty$. 
  Fix $\eta > 0$ and $\delta > 0$. Throughout, we suppress dependence of quantities on $N$. By Assumption~\ref{assumption:5-tightness} applied with $\epsilon = \eta \delta / 4$ there exists a neighborhood $U$ of $x$ and a compact set $K \subseteq \sY$ such that
  \begin{math}
    \sup_{z \in U} F_z(K^c) \le \epsilon.
  \end{math}
  Now, write
  \begin{align*}
    \Fhat_x(K^c) = \sum_{i =1 }^N w_i(x) 1(Y_i \notin K)
    \le R(U) + S(U, K)
  \end{align*}
  where $R(U) = \sum_i w_i(x) 1(X_i \notin U)$ and $S(U, K) = \sum_i w_i(x) 1(X_i \in U, Y_i \notin K)$. Consequently,
  \begin{align}
    \label{eq:target}
    \Pr\{\Fhat_x(K^c) > \eta\} \le \Pr\{R(U) > \eta / 2\} + \Pr\{S(U, K) > \eta / 2\}.
  \end{align}
  We first show that there exists an $N^\star$ such that for $N \ge N^\star$ the left hand side of \eqref{eq:target} is less than $\delta$. First, by Lemma~\ref{lemma:weight-localization}, we have $\Pr\{R(U) > \eta / 2\} \to 0$ as $N \to \infty$; take $N^\star$ so that $\Pr\{R(U) > \eta / 2\} \le \delta / 2$ for all $N \ge N^\star$. For the second term, we apply Markov's inequality to obtain
  \begin{align}
    \label{eq:lem1-S}
    \Pr\{S(U, K) > \eta / 2\} \le 2 \E\{S(U, K)\} / \eta.
  \end{align}
  This is given by
  \begin{align*}
    \sum_i \E\{w_i(x) \, 1(X_i \in U) \, 1(Y_i \notin K)\}
    &= 
    \sum_i \E\{w_{m,i}(x) \, 1(X_i \in U) \, 1(Y_i \notin K)\}
    \\&= 
    \sum_i \E\{w_i(x) \, 1(X_i \in U) \, \Pr(Y_i \notin K \mid X_i)\}
    \\&\le
    \{\sup_{z \in U} F_{z}(K^c)\} \E \sum_i w_{m,i}(x) 1(X_i \in U)
    \\&\le
    \sup_{z \in U} F_{z}(K^c)
    \le
    \eta \delta / 4,
  \end{align*}
  where the second equality holds by the honesty assumption, which states that $w_{m,i}(x)$ is independent of $Y_i$ given $X_i$ whenever $w_{m,i}(x) \ne 0$. Combining this with \eqref{eq:lem1-S} gives $\Pr\{S(U, K) > \eta / 2\} \le \delta / 2$. This establishes $\Pr\{\Fhat_x(K^c) > \eta\} \le \delta$ for $N \ge N^\star$.

  This leaves the initial segment $N = 1,\ldots,N^\star$ to bound. Fix $N$ in this range and let $K_N$ be a compact set to be specified later. By Markov's inequality we have
  \begin{align*}
    \Pr\{\Fhat_x(K_N^c) > \eta\} \le \frac{\E\{\Fhat_x(K_N^c)\}}{\eta}
    = \frac{\bar{F}_{N,x}(K_N^c)}{\eta}
  \end{align*}
  where $\bar{F}_{N,x}(\cdot)$ is the expected value. By Theorem~1.3 of \citet{billingsley2013convergence}, $\bar{F}_{N,x} \in \sP(\sY)$ is tight because $\sP(\sY)$ (or, for the second part of the result, $\sP(C)$ under $\|\cdot\|$) is a Polish space, and hence $K_N$ can be chosen to be a compact set satisfying $\bar{F}_{N,x}(K_N^c) \le \delta \eta$. Hence, with that choice, we have $\Pr\{\Fhat_x(K_N^c) > \eta\} \le \delta$ for $N \le N^\star$.
  Combining the finite segment with the tail segment, we can take our compact set to be $K^\star = \bigcup_{N = 1}^{N^\star} K_N \cup K$.
\end{proof}

\begin{lemma}
  \label{lemma:mmd-continuity}
  Suppose that $\kappa$ is bounded and continuous on $\sY \times \sY$. Then the mapping $Q_x : P \mapsto \MMD(P, F_x)$ is continuous with respect to the weak topology on $\sP(\sY)$. The same result also holds with $(\sY, \|\cdot\|_{\sY})$ replaced with $(C, \|\cdot\|_\infty)$.
\end{lemma}

\begin{proof}
  It suffices to prove continuity of $Q_x(\cdot)^2$. For any $P \in \sP(\sY)$ we have 
  \begin{align*}
    Q_x(P)^2 = \iint \kappa(y, y') \ P(dy) \, P(dy')
    - 2 \iint \kappa(y, y') \ P(dy) \, F_x(dy')
    + \iint \kappa(y, y') \ F_x(dy) \, F_x(dy').
  \end{align*}
  Now, suppose $F_N \to P$ weakly in $\sP(\sY)$. Because $\sY$ is a separable Banach space, $F_N \otimes F_N \to P \otimes P$ and $F_N \otimes F_x \to P \otimes F_x$ weakly on $\sY \times \sY$. Because $\kappa$ is bounded and continuous on $\sY \times \sY$, the integrals above evaluated at $Q_x(F_N)^2$ converge so that $Q_x(F_N)^2 \to Q_x(P)^2$. The argument with $(\sY, \|\cdot\|_{\sY})$ replaced with $(C, \|\cdot\|_\infty)$ is identical.
\end{proof}

\begin{lemma}[Tightness of the Laws]
  \label{lemma:tightness-2}
  Suppose that Assumptions~\ref{assumption:1-rf}---\ref{assumption:5-tightness} hold. Let $G_N \in \sP\{\sP(\sY)\}$ denote the law of $\Fhat_{x}$. Then the sequence of laws $\{G_N\}_{N = 1}^\infty$ is tight on $\sP(\sY)$ endowed with the weak topology. That is, for every $\delta > 0$ there exists a compact set $\sK \subseteq \sP(\sY)$ such that 
  \begin{align*}
    \limsup_{N \to \infty} G_N(\sK^c) \le \delta.
  \end{align*}
  If we additionally assume Assumption~\ref{assumption:6-ctightness} then this result holds for some $\sK$ chosen compact in $\sP(C)$ under the topology induced by $\|\cdot\|_\infty$ and $G_N$ interpreted as being in $\sP\{\sP(C)\}$.
\end{lemma}

\begin{proof}
  We first proceed under Assumption~\ref{assumption:1-rf}---\ref{assumption:5-tightness}. Fix $\delta > 0$ and apply Lemma~\ref{lemma:tightness-1} with $\eta = 2^{-m}$ for each $m \ge 1$ and probability tolerance $\delta 2^{-m}$. This gives compact sets $K_m$ such that
  \begin{align*}
    \sup_N \Pr\{\Fhat_x(K_m^c) > 2^{-m}\} \le \delta 2^{-m}.
  \end{align*}
  Define the set
  \begin{align*}
    \sK = \bigcap_{m = 1}^\infty \left\{ P \in \sP(\sY) : P(K_m) \ge 1 - 2^{-m} \right\}.
  \end{align*}
  We will show that $\sK$ is compact and that $\limsup_{N \to \infty} G_N(\sK^c) \le \delta$. First, by the Portmanteau theorem (\citealp{billingsley2013convergence} Theorem~2.1) the mapping $P \mapsto P(K_m)$ is upper semi-continuous because $K_m$ is closed; hence each set $\{P \in \sP(\sY) : P(K_m) \ge 1 - 2^{-m}\}$ is closed and therefore $\sK$ is closed.

  By Prohorov's theorem (\citealp{billingsley2013convergence} Theorem~5.1), compactness of $\sK$ follows from uniform tightness of $\sK$. To show this, given $\epsilon > 0$ we choose $K_m$ as our compact set with $m$ such that $2^{-m} < \epsilon$. Then every $P \in \sK$ satisfies $P(K_m^c) < \epsilon$ by the definition of $\sK$.

  It remains to show that $\limsup_{N \to \infty} G_N(\sK^c) \le \delta$. But
  \begin{align*}
    \limsup_{N \to \infty} G_N(\sK^c)
    &= \limsup_{N \to \infty} G_N\left( \bigcup_{m \ge 1} \{\Fhat_x(K_m^c) > 2^{-m}\}  \right)
    \\&
    \le \sum_{m = 1}^\infty \sup_{N} G_N\{\Fhat_x(K_m^c) > 2^{-m}\}
    \\&
    \le \sum_{m = 1}^{\infty} \delta 2^{-m} = \delta.
  \end{align*}
  Because $\delta > 0$ was arbitrary, this proves tightness of the laws $\{G_N\}$ on $\sP(\sY)$.

  Under Assumption~\ref{assumption:6-ctightness} (ii), the result is immediate as all the sets $\sK$ are compact with respect to $C$ under $\|\cdot\|_\infty$. Under Assumption~\ref{assumption:6-ctightness} (i), the argument proceeds in exactly the same fashion by taking the $K_m$'s to be compact in $C([0,1])$ using Lemma~\ref{lemma:tightness-1} and taking $\sK = \bigcap_{m = 1}^\infty \{P \in \sP(C([0,1])) : P(K_m) \ge 1 - 2^{-m}\}$ and concluding tightness of the laws $\{G_N\}$ on $\sP(C([0,1]))$.
\end{proof}

\begin{theorem}[Weak consistency in probability]
  Suppose that Assumptions~\ref{assumption:1-rf}---\ref{assumption:5-tightness} hold. Then $\Fhat_x \to F_x$ weakly-in-probability as random elements in $\sP(\sY)$. If Assumption~\ref{assumption:6-ctightness} holds as well, then $\Fhat_x \to F_x$ weakly-in-probability as random elements in $\sP(C[0,1])$.
\end{theorem}

\begin{proof}
  We prove the first claim. By Lemma~\ref{lemma:tightness-2}, the laws $G_N$ form a tight family on the Polish space $\sP(\sY)$. Let $N_r$ be an arbitrary subsequence. 
  By tightness, there exists a further subsequence $N_{r_j}$ and a $\sP(\sY)$-valued random element $Q$ such that $G_{N_{r_j}}$ converges weakly to the law of $Q$ (Theorem~5.1 of \citealp{billingsley2013convergence}). By Lemma~\ref{lemma:mmd-continuity} and the continuous mapping theorem, $\DMMD(\Fhat_{x}, F_x) \cind \DMMD(Q, F_x)$ along the subsequence $N_{r_j}$. On the other hand, Theorem~\ref{thm:mmd} states that $\DMMD(\Fhat_x, F_x) \cinp 0$; in particular, this holds along $N_{r_j}$, and so $\DMMD(Q, F_x) = 0$ almost surely. Because $\kappa$ is assumed to be characteristic, we have $Q = F_x$ almost surely; hence, for every subsequence $N_r$ there exists a further subsequence $N_{r_j}$ such that $G_{N_{r_j}} \to \delta_{F_x}$ weakly. By the Urysohn subsequence principle (Theorem~2.6 of \citealp{billingsley2013convergence}), it follows that $G_{N} \to \delta_{F_x}$ weakly.

  The argument for the second claim in case (i) of Assumption~\ref{assumption:6-ctightness} is identical, with $\sY$ replaced by $C([0,1])$ throughout. For case (ii), because $\sY$ is continuously embeddable into $C([0,1])$, the result follows from the first claim and the continuous mapping theorem.
\end{proof}

\section{Data Generating Mechanism for Section~\ref{sec:manifold-faithfulness-experiment}}
\label{sec:appendix-dgp}

The functional response was generated on an equally spaced grid $t_1,\ldots,t_G \in [0,1]$ with $G=101$, and each $Y_i$ was normalized to be a density. For each subject we sampled covariates
\[
  X_i = (X_{i1}, X_{i2}), \qquad
  X_{i1} \sim \operatorname{Unif}(-1,1), \quad
  X_{i2} \sim \Normal(0,1),
\]
We first draw a latent mixture indicator
\begin{math}
  Z_i \sim \operatorname{Bernoulli}\{p(X_{i1})\}
  \text{ and }
  p(x) = \{1+\exp(2.2x)\}^{-1},
\end{math}
where $Z_i=1$ corresponds to a ``morning peak'' in physical activity. The baseline peak location was $c_i=0.32$ if $Z_i=1$ and $c_i=0.68$ otherwise. The realized peak location and width were then
\[
  \mu_i =
  \Pi_{[0.15,0.85]}
  \left\{
    c_i + \epsilon_{\mu i}
  \right\},
  \qquad
  \epsilon_{\mu i} \sim
  \Normal\left(0,\,
    \left[0.015 + 0.055\{1+\exp(-2X_{i1})\}^{-1}\right]^2
  \right),
\]
and
\[
  \sigma_i =
  \Pi_{[0.035,0.11]}
  \left[
    \exp\{\log(0.055) + 0.25X_{i1} + \epsilon_{\sigma i}\}
  \right],
  \qquad
  \epsilon_{\sigma i} \sim \Normal(0,0.10^2),
\]
where $\Pi_A$ denotes projection onto the interval $A$. Given $(\mu_i,\sigma_i)$,
we formed the unnormalized profile
\[
  q_i(t)
  =
  \left[
    \exp\left\{-\frac{1}{2}\left(\frac{t-\mu_i}{\sigma_i}\right)^2\right\}
    +
    0.25
    \exp\left\{-\frac{1}{2}
      \left(\frac{t-\min(\mu_i+0.08,0.95)}{1.4\sigma_i}\right)^2
    \right\}
    + 0.02
  \right]
  \exp\{W_i(t)\}.
\]
The smooth perturbation $W_i$ was generated from the finite Fourier expansion
\[
  W_i(t)
  =
  \sum_{k=1}^6 a_{ik}\sin(2\pi kt)
  +
  \sum_{k=1}^6 b_{ik}\cos(2\pi kt),
\]
with independent coefficients
\[
  a_{ik}, b_{ik}
  \sim
  \Normal\left(
    0,\,
    \left[
      \frac{0.06 + 0.08\{1+\exp(-X_{i1})\}^{-1}}{k^{3/2}}
    \right]^2
  \right).
\]
Finally, the observed curve was
\[
  Y_i(t)
  =
  \frac{\max\{q_i(t), 0\}}{\int_0^1 \max\{q_i(u), 0\} \ du}
\]
Thus $X_{i1}$ changes the conditional law through the morning/evening mixture probability, the amount of phase variation, the width of the peak, and the magnitude of smooth multiplicative shape perturbations.

\section{More on Kernel Choice}

The kernel used by the \fdrf\ should be viewed as part of the statistical specification of the problem of interest. In particular, the kernel determines which aspects of the conditional law are made visible to the splitting rule. An $L_2$ Gaussian kernel gives a natural baseline because it compares curves through their overall level and shape, but it can underweight differences that are scientifically important yet small in $L_2$, such as changes in roughness, local oscillation, or derivative behavior. Transform kernels of the form \eqref{eq:wynne} provide a simple way to target these features: rather than comparing $y$ and $y'$ directly, one compares $T(y)$ and $T(y')$ in a Hilbert space chosen to encode the desired geometry. When $T$ is injective and sufficiently regular, this preserves the characteristic property of the Gaussian kernel, while allowing the analyst to emphasize features of the response that are most relevant for the application.

\begin{table}[t]
\small
\centering
\begin{tabularx}{\textwidth}{>{\raggedright\arraybackslash}p{0.18\textwidth} >{\raggedright\arraybackslash}p{0.29\textwidth} X}
\toprule
\textbf{Option} & \textbf{Kernel / transform} & \textbf{When to use it} \\
\midrule
Identity / $L_2$ &
$T_{\rm id}(u)=u$. &
Good default for global level and shape changes.  Characteristic on the usual Hilbert-space setup when the identity map is admissible.  It can be insensitive to roughness or covariance changes that have small $L_2$ signal. \\
\addlinespace
Sobolev / derivative &
$T_\lambda(u)=(\sqrt\lambda u,\sqrt{1-\lambda}\,Du)$, equivalently $\|u-v\|_{H^1_\lambda}^2=\lambda\|u-v\|_2^2+(1-\lambda)\|D(u-v)\|_2^2$. &
Strong candidate for the main \fdrf\ analysis.  It detects shifts in level, peak timing, smoothness, and roughness.  Including the $u$ component preserves injectivity; choosing $\lambda\approx 1/2$ after scale normalization is a reasonable default. \\
\addlinespace
FPCA / low-rank operator &
$T_F(u)=\sum_{r=1}^F a_r\langle u,e_r\rangle_2 e_r$, with empirical eigenfunctions $e_r$ and weights $a_r$, often $a_r=\lambda_r^{1/2}$. &
Useful when curves are noisy or sparsely observed and most relevant variation is low-dimensional.  Choose $F$ by variance explained, e.g. 90--95\%, or by held-out score.  Because finite $F$ discards directions, this is a pragmatic non-characteristic kernel and should be reported as a sensitivity analysis. \\
\addlinespace
Square / moment expansion &
$T_{\rm sqr}(u)=(u,u^2)$, with separate bandwidths for $u$ and $u^2$: $\exp\{-\|u-v\|_2^2/(2\gamma_1^2)-\|u^2-v^2\|_2^2/(2\gamma_2^2)\}$. &
Targets heteroskedasticity, amplitude variation, and second-order differences that a mean-sensitive kernel may miss.  Requires enough integrability, e.g. $u\in L_4$, and benefits from centering/scaling before squaring. \\
\addlinespace
Smoothed spectral / operator &
$T(u)=C^{1/2}u$ or $C_{k_0}u$, where $C$ downweights high frequencies.  This includes cosine-exponential or other smoothing operators. &
Useful when the relevant signal is smooth and high-frequency noise would dominate $L_2$.  If infinitely many positive spectral weights are retained, $T$ can remain injective; finite truncations trade characteristicness for stability. \\
\addlinespace
Integral / random-projection &
For $C\in\mathcal L_1^+(\sY)$ and scalar kernel $k_0$, set $\kappa_{C,k_0}(u,v)=\int k_0(\langle u,h\rangle,\langle v,h\rangle)d\mathcal N_C(h)$.  If $k_0(a,b)=\cos(a-b)$, this equals $\exp\{-\langle C(u-v),u-v\rangle/2\}$. &
Appealing for dense grids, irregular observations, or computational shortcuts via random features.  It compares distributions through many random linear probes of the function.  With an injective $C$ and suitable translation-invariant $k_0$, it can be characteristic. \\
\bottomrule
\end{tabularx}
\caption{Summary of possible choices of kernel function for building the \fdrf.\label{tab:kernels}}
\end{table}

In Table~\ref{tab:kernels}, we list out some options for $T(\cdot)$, with most of these also suggested by \citet{wynne2022kernel}, who also studies the impact of the choice of transformation on the statistical power in the context of two-sample testing. Transforms, such as Sobolev transforms or moment expansions, can make the forest responsive to distributional changes in smoothness, timing, or scale that would be muted under the identity transform. On the other hand, low-rank or smoothed transforms, such as FPCA-based or operator-based kernels, may improve stability when curves are noisy, sparsely observed, or contain high-frequency measurement error; finite truncations also generally sacrifice characteristicness because some directions in the function space are discarded. For this reason, we find the Sobolev kernel to be a strong default when derivative information is meaningful, and also recommend both reporting the $L_2$ kernel result as a baseline and using more specialized kernels as sensitivity analyses tied to specific scientific questions.

\subsection{Random Features for Non-Gaussian Radial Kernels}
\label{sec:non-gaussian-random-features}

The random-feature representation is not specific to Gaussian kernel, but can also be applied to a broad a broad class of radial kernels. Let $T:\sY\to\sG$ map the response into a real separable Hilbert space. If
\begin{align}
  \label{eq:scale-mixture-kernel}
  \kappa_S(y,y')
  =
  \E_S\left[
    \exp\left\{-\frac{S \|y - y'\|_2^2}{2\ell^2}\right\}
  \right]
\end{align}
for a nonnegative mixing variable $S$, then the usual cosine features remain valid after multiplying the Gaussian projection by $\sqrt{S}$ \citep{rahimi2007random,langrene2026spectral}.

\citet{rahimi2007random} list several kernels htat have simple mixing laws. The rational-quadratic kernel $\kappa_{RQ}(y,y') = \{1 + \|y - y'\|_2^2 / (2\alpha \ell^2)\}^{-1}$ can be obtaiend by taking $S \sim \Gam(\alpha, \alpha)$.
The rational-quadratic mixture is a standard option \citep[Section~4.2]{rasmussen2006gaussian}, and the Cauchy kernel is its $\alpha = 1$ case. The Laplace kernel can be obtained by taking $S = Z_0^{-2}$ with $Z_0 \sim \Normal(0,1)$.

To apply these alternate kernels, for a quadrature grid $(t_r,\Delta_r)_{r=1}^R$ in $\sG=L_2([0,1])$, we draw $S_j$ from the appropriate law above, $Z_{rj}\iid\Normal(0,1)$, and $b_j\sim\Uniform(0,2\pi)$, then use
\begin{align}
  \xi_{rj}=\sqrt{S_j}Z_{rj},
  \qquad
  \varphi_j(y)
  =
  \sqrt{2}\cos\left\{
    \frac{1}{\ell}
    \sum_{r=1}^R \sqrt{\Delta_r}\,\xi_{rj}[T(y)](t_r)
    +b_j
  \right\}.
  \label{eq:scale-mixture-grid-random-feature}
\end{align}
For an equally spaced grid, $\sqrt{\Delta_r} = R^{-1/2}$. Product-space transforms are handled by concatenating their weighted coordinates. The shared $S_j$ makes the coefficients marginally dependent and should be sampled once per feature. Their averaged squared empirical mean differences give the MMD for \eqref{eq:scale-mixture-kernel} by the standard random-feature argument \citep{rahimi2007random}.

\begin{figure}[p]
  \centering
  \includegraphics{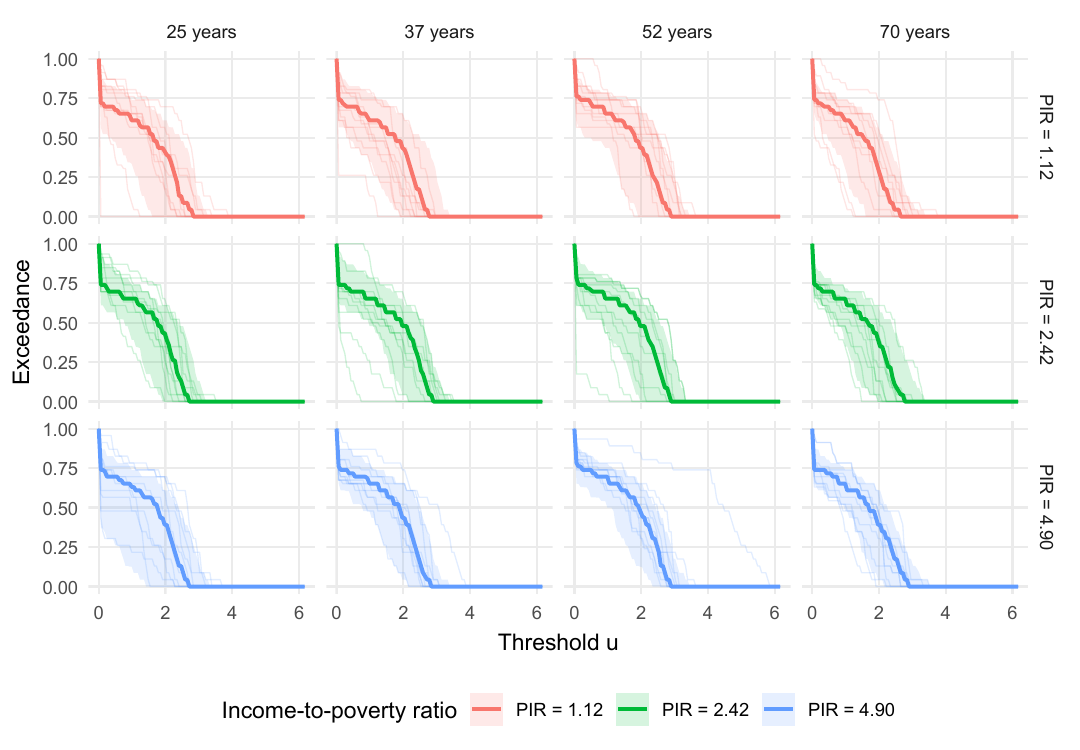}
  \caption{Predictive distribution of the exceedance function in different strata as a function of the activity threshold $u$.}
  \label{fig:exceed}
\end{figure}

\section{Algorithms}
\label{sec:appendix-algorithms}

This section collects the implementation details behind the \fdrf. Algorithm~\ref{alg:fdrf-fit-predict} describes how to implement the method as described in Sections~\ref{sec:random-forest} and~\ref{sec:estimating-conditional-distributions}: grow many honest trees, drop a new covariate value through each tree, average the terminal-leaf weights, and return a weighted empirical distribution over the observed functions. Algorithm~\ref{alg:fdrf-grow-tree} describes how a single tree is grown. 

\begin{algorithm}[t]
\caption{Fitting and predicting with the functional distributional random forest\label{alg:fdrf-fit-predict}}
\begingroup
\footnotesize
\begin{enumerate}
\setlength{\itemsep}{0.15em}
\setlength{\parsep}{0pt}
\setlength{\parskip}{0pt}

\item
\textbf{Inputs.}
Data $\{(X_i,Y_i)\}_{i=1}^N$, number of trees $M$, subsample size $s_N$,
honesty fraction $\rho$, candidate-variable parameter $m_{\rm try}$, minimum
node size $n_{\min}$, balance parameter $\alpha$, kernel $\kappa$, transform
$T:\sY\to\sG$, and number of random features $B$.

\item
\textbf{Preprocess.}
Compute $T(Y_i)$ for $i=1,\ldots,N$. If the length scale is not supplied, set
$\ell=\median\{\|T(Y_i)-T(Y_j)\|_{\sG}:i<j\}$.

\item
\textbf{Grow trees.}
For each $m=1,\ldots,M$, draw $S_m\subseteq\{1,\ldots,N\}$ without replacement,
with $|S_m|=s_N$. For each tree, split $S_m$ into an honest split sample
$S_m^{\rm split}$ and estimation sample $S_m^{\rm est}$, with
$|S_m^{\rm split}|=\lfloor\rho s_N\rfloor$. Starting from the root cell
$R_{b_0}=\sX$, grow a tree using Algorithm~\ref{alg:fdrf-grow-tree}, using
only $S_m^{\rm split}$ to choose splits.

\item
\textbf{Populate terminal leaves.}
After the tree structure is fixed, attach the honest estimation observations to
the leaves. That is, for each terminal leaf $\ell\in\Leaves(\Tree_m)$, store
$\{i\in S_m^{\rm est}:X_i\in R_\ell\}$.

\item
\textbf{Predict at \(x\).}
Let $\ell_m(x)$ be the leaf of tree $m$ containing $x$, and define
$I_m(x)=\{i\in S_m^{\rm est}:X_i\in R_{\ell_m(x)}\}$ and
$n_m(x)=|I_m(x)|$. Set
\[
  w_{m,i}(x)=\frac{1\{i\in I_m(x)\}}{n_m(x)},
  \qquad
  w_i(x)=\frac{1}{M}\sum_{m=1}^M w_{m,i}(x).
\]

\item
\textbf{Return the conditional law.}
Return
$\Fhat_x(\cdot)=\sum_{i=1}^N w_i(x)\delta_{Y_i}(\cdot)$.

\end{enumerate}
\endgroup
\end{algorithm}

\begin{algorithm}[t]
\caption{Growing one \fdrf\ tree}
\label{alg:fdrf-grow-tree}
\begingroup
\footnotesize
\begin{enumerate}
\setlength{\itemsep}{0.15em}
\setlength{\parsep}{0pt}
\setlength{\parskip}{0pt}

\item
\textbf{Initialize a node.}
At node $b$, let $R_b$ be the current cell,
$I_b=\{i\in S_m^{\rm split}:X_i\in R_b\}$, $n_b=|I_b|$, and
$Y_b=\{Y_i:i\in I_b\}$.

\item
\textbf{Select Variables.}
Draw $q_b=\min\{\max[\operatorname{Poisson}(m_{\rm try}),1],P\}$ candidate coordinates and let $J_b\subseteq\{1,\ldots,P\}$ be the selected set.

\item
\textbf{Construct the candidate split set.}
For $j\in J_b$, consider cutpoints $c$ between consecutive distinct values of
$X_{ij}$ among $i\in I_b$. Let $\mathcal A_b$ contain the splits satisfying the $\alpha$-regularity condition.

\item
\textbf{Score candidate splits.}
For each $(j,c)\in\mathcal A_b$, define
$R_{bL}(j,c)$, $R_{bR}(j,c)$, $Y_{bL}(j,c)$, and $Y_{bR}(j,c)$ as in
Section~\ref{sec:random-forest}. Choose $(j_b,c_b)$ by the MMD split rule
\eqref{eq:fdrf-split}. When exact MMD is computationally infeasible, replace
$\DMMD$ in \eqref{eq:fdrf-split} by the random-feature approximation based on
$\{\varphi_{W_j,b_j}\}_{j=1}^B$ in \eqref{eq:random-features}.

\item
\textbf{Split or stop.}
If $\mathcal A_b$ is empty, declare $b$ terminal. Otherwise, store
$(j_b,c_b)$, set
$R_{bL}=R_b\cap\{x:x_{j_b}\le c_b\}$ and
$R_{bR}=R_b\cap\{x:x_{j_b}>c_b\}$, and recursively apply the same steps to the
two child nodes.

\end{enumerate}
\endgroup
\end{algorithm}

Algorithm~\ref{alg:fdrf-confidence-intervals} gives the little-bags procedure used for uncertainty quantification for plug-in functionals of the conditional law, as described in Section~\ref{sec:inference-target-functionals}. This approach uses the half-sampling strategy for computing standard errors. We use an honesty fraction of $\rho = 1/2$ by default as well. For details on the validity of this approach, see \citet{naf2023confidence} and \citet{sexton2009standard}. A variant of this strategy is also used by \citet{athey2019generalized} to estimate a component of the sandwich matrix variance approximation for generalized random forests.

We again note that the standard errors produced by Algorithm~\ref{alg:fdrf-confidence-intervals} overstate the sampling variability due to not accounting for within-bag variability. We describe how to account for this if desired, however we again note that empirically we found this to reduce the coverage of intervals substantially and routinely induce pathologies such as negative variance estimates. A simple approach is to approximate $\widehat V_{\texttt{between}}(x) \approx \widehat V_{\infty}(x) + B(x) / M_{\texttt{bag}}$ where $B(x) / M_{\texttt{bag}}$ is the within-bag variance. We can approximate $\widehat V_{\infty}(x)$ by fitting a second forest with a bag size of $2 M_{\texttt{bag}}$ and setting
\begin{math}
  \widehat V_{\infty}(x) \approx 2V_{\texttt{between}}(2M_{\texttt{bag}}) - V_{\texttt{between}}(M).
\end{math}
We can then take $\sqrt{\widehat V_\infty(x)}$ as an approximation of the standard error. This approach can be improved further to account for the fact that the total number of trees is also finite by adding the correction $\widehat B(x) / T$ where $\widehat B(x) = 2 M \{\widehat V_{\texttt{between}}(M) - \widehat V_{\texttt{between}}(2M)\}$. To address negative variance estimates, \citet{athey2019generalized} recommend using a Bayesian shrinkage estimator to ensure positivity.

\begin{algorithm}[t]
\caption{Pointwise confidence intervals for \fdrf\ functionals}
\label{alg:fdrf-confidence-intervals}
\begingroup
\footnotesize
\begin{enumerate}
\setlength{\itemsep}{0.15em}
\setlength{\parsep}{0pt}
\setlength{\parskip}{0pt}

\item \textbf{Inputs.} Training data $\{(X_i,Y_i)\}_{i=1}^N$, target covariate values $x_1,\ldots,x_K$, a scalar functional $\theta:\sP(\sY)\to\Reals$, confidence level $1-\alpha$, and little-bag parameters $(L, M)$.

\item \textbf{Grow little bags.} For each bag $b=1,\ldots,L$ construct $M$ honest \fdrf\ trees based on a half-sample. Let $\Fhat^{(b,t)}_{x_k}$ denote the tree-level conditional law at $x_k$ from tree $t$ in bag $b$.

\item
\textbf{Compute full-forest estimates.}
For each target value $x_k$, average all tree-level laws across all bags,
\begin{math}
  \Fhat_{x_k}
  =
  \frac{1}{L M}
  \sum_{b=1}^{L}
  \sum_{t=1}^{M}
  \Fhat^{(b,t)}_{x_k},
\end{math}
and set
\begin{math}
  \widehat\theta(x_k)=\theta(\Fhat_{x_k}).
\end{math}

\item
\textbf{Compute bag-level estimates.}
For each bag $b=1,\ldots,L$ and target value $x_k$, form
\begin{math}
  \Fhat^{(b)}_{x_k}
  =
  \frac{1}{M}
  \sum_{t=1}^{M}
  \Fhat^{(b,t)}_{x_k},
  \
  \widehat\theta^{(b)}(x_k)
  =
  \theta(\Fhat^{(b)}_{x_k}).
\end{math}
\item
\textbf{Compute between-bag variability.}
Let
\begin{math}
  \widetilde\theta(x_k)
  =
  \frac{1}{L}
  \sum_{b=1}^{L}
  \widehat\theta^{(b)}(x_k).
\end{math}
Compute
\[
  \widehat V(x_k)
  =
  \frac{1}{L-1}
  \sum_{b=1}^{L}
  \left\{
    \widehat\theta^{(b)}(x_k)-\widetilde\theta(x_k)
  \right\}^2 .
\]

\item
\textbf{Return pointwise confidence intervals.}
Return
\begin{math}
  \widehat\theta(x_k)
  \pm
  z_{1-\alpha/2}
  \sqrt{\widehat V(x_k)}
\end{math}, for
\begin{math}
  k=1,\ldots,K
\end{math}.

\end{enumerate}
\endgroup
\end{algorithm}

\section{Additional Simulation Results}

\paragraph{Visualization of Conditional Distribution Estimates}

To visualize the accuracy of the \fdrf\ for the simulation settings in Section~\ref{sec:simulation-study}, a single dataset was generated from each of the DGPs. Conditional distributions from both the true distribution and the fitted model are given in Figure~\ref{fig:dgm-faceted} and Figure~\ref{fig:sobolev-drf-faceted}, respectively. To provide a concrete predictive distribution for the \fdrf\ to target, we fixed $X_i$ so that $X_i^\top\beta = \eta$ by taking the entries of $X_i$ to be constant.  We see that the \fdrf\ is capable of accurately capturing the variability in the predictive distribution across all settings, and closely matches the shape and spread of the true predictive distributions. It is also capable of capturing varying degrees of smoothness, spread, and so forth.

\begin{figure}[p]
  \centering
  \hspace*{-4em}
  \includegraphics[width=.9\textwidth]{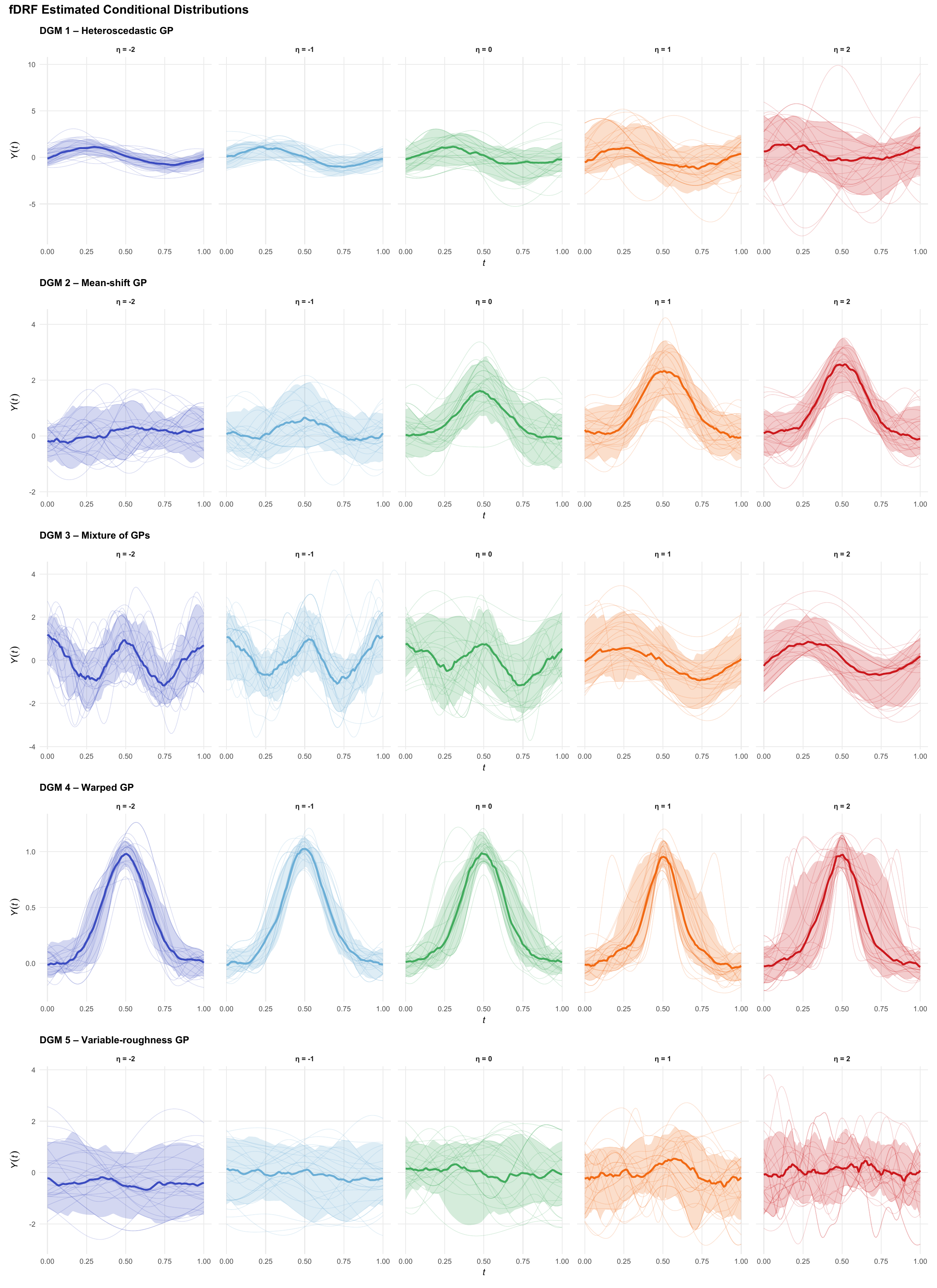}
  \caption{Estimated conditional distributions from the \fdrf\ for a single simulated dataset for each of the data generating mechanisms.}
  \label{fig:sobolev-drf-faceted}
\end{figure}

\paragraph{Results Using $L_2$-Based Scoring Rules}

We repeated the comparison in Figure~\ref{fig:simulation} using $T(y)=y$ rather than $T(y)=(y,Dy)$ to construct the energy and Gaussian kernel scores. The data-generating mechanisms, fitted methods, and other simulation settings were unchanged.

Results are summarized in Figure~\ref{fig:simulation-l2}. The conclusions for the heteroskedastic, mixture, and warped Gaussian-process settings are qualitatively similar to those obtained using the Sobolev-based scores. The \fdrf\ methods remain among the leading methods in these settings, although the pairwise score differences are generally smaller and fewer pairwise comparisons receive an $(H,M,L)$ tag. The three \fdrf\ variants also perform more similarly, with the FPCA-based \fdrf\ becoming relatively more competitive in the mixture and warped settings.

The most substantial differences occur in the mean-shift and roughness settings. For the mean-shift Gaussian process, the functional linear model has the best average rank under both $L_2$-based scores, and the pairwise comparisons favor it over each of the \fdrf\ variants. This is consistent with the fact that the distributional change in this setting occurs through the conditional mean, which the functional linear model can approximate well. For the Roughness DGP, the strong advantage of the Sobolev-based \fdrf\ under the Sobolev scores disappears almost entirely. Although the functional linear model has the best average rank under the $L_2$-based scores, the score differences are small and none of the pairwise comparisons receives an $(H,M,L)$ tag. We therefore interpret this result as a failure of the $L_2$-based scores to distinguish methods that capture conditional roughness from those that do not, rather than as evidence that the functional linear model accurately captures the changes in roughness.

Overall, using the $L_2$ norm makes the fitted methods appear more similar and reduces the apparent improvement of the \fdrf\ over the functional linear model. This comparison provides further evidence that Sobolev-based scoring rules are preferable when differences in smoothness and roughness are important features of the functional distribution.

\begin{figure}[t]
  \centering
  \includegraphics[width=.9\textwidth]{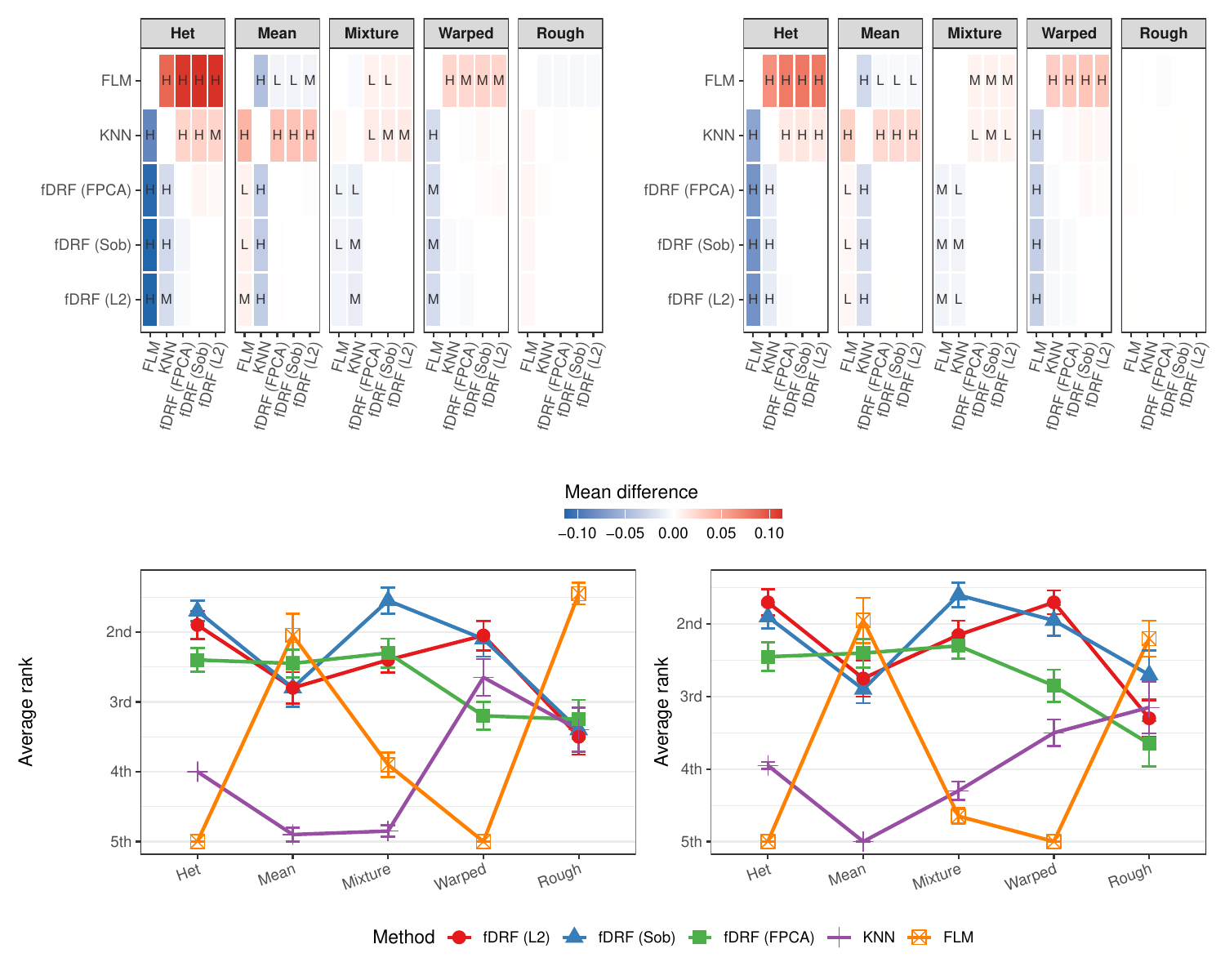}
  \caption{Simulation results using $L_2$-based scoring rules. Top: pairwise comparison heatmaps under the $L_2$ energy score (left) and the Gaussian kernel score constructed from the $L_2$ norm (right). Blue indicates that the row method performed better, while red indicates that the column method performed better; cells marked $(H,M,L)$ have median paired-test $P$-values below $0.001$, $0.01$, and $0.05$, respectively. Bottom: average ranking of each method in each simulation setting under the $L_2$ energy score (left) and Gaussian kernel score (right). Error bars show the mean rank $\pm$ one standard error across the 20 replicated datasets.}
  \label{fig:simulation-l2}
\end{figure}

\section{Additional Analysis of NHANES}

In this section we provide additional analyses for the NHANES dataset. In the main text, Figure~\ref{fig:nhanes-peak-effort-functional-intervals} summarizes peak time and maximal effort through four scalar functionals: the circular mean and circular standard deviation of $P_i$ and the mean and standard deviation of $M_i=\max_t Y_i(t)$. The full predictive distributions for $P_i$ and $M_i$ are given in Figure~\ref{fig:nhanes-peak-effort-predictive-violins}. We see that differences in the predictive distributions across ages and income strata are quite subtle for both $M_i$ and $P_i$, indicating that even within the strata of white/male/married high-school graduates there is substantial variability in these quantities.

Figure~\ref{fig:nhanes-raw-trajectory-median} complements Figure~\ref{fig:nhanes-age-income-distribution} and Figure~\ref{fig:nhanes-raw-trajectory-age-contrasts} by providing pointwise confidence intervals for the median of the day-averaged activity trajectory. We see clearly that, as income rises, the median profile for older individuals tends to increase overall while the younger individuals have median profiles that are relatively stable across income.

\begin{figure}[t]
  \centering
  \includegraphics[width=\textwidth]{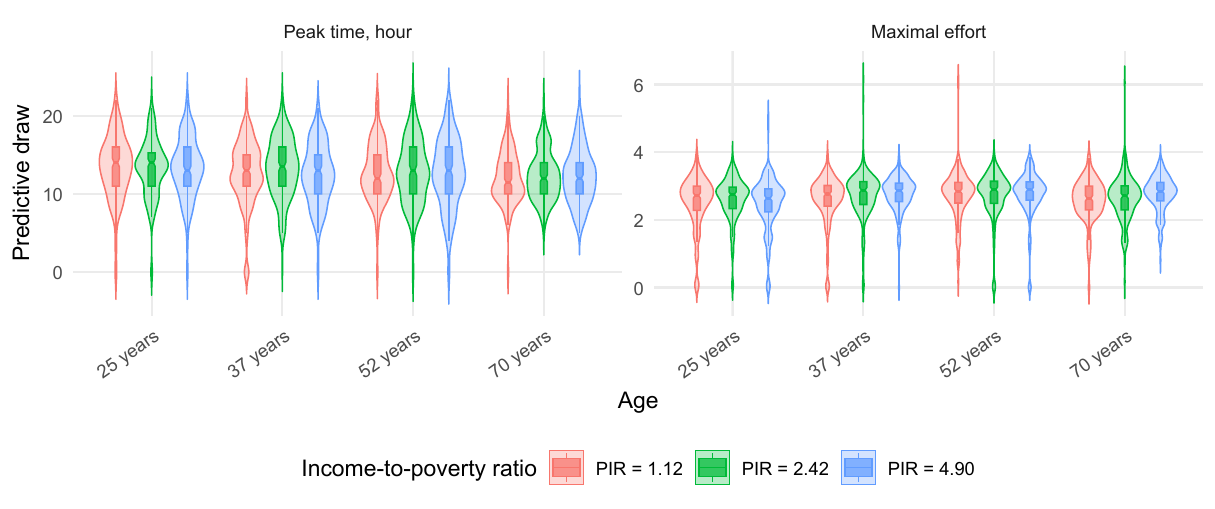}
  \caption{Predictive distribution of the peak time (left) and maximal effort (right) across strata, displayed with boxplots and violin plots. Boxplots give upper and lower quartiles and the median, with whiskers extending up to 1.5 times the interquartile range.}
  \label{fig:nhanes-peak-effort-predictive-violins}
\end{figure}

\begin{figure}[t]
  \centering
  \includegraphics[width=\textwidth]{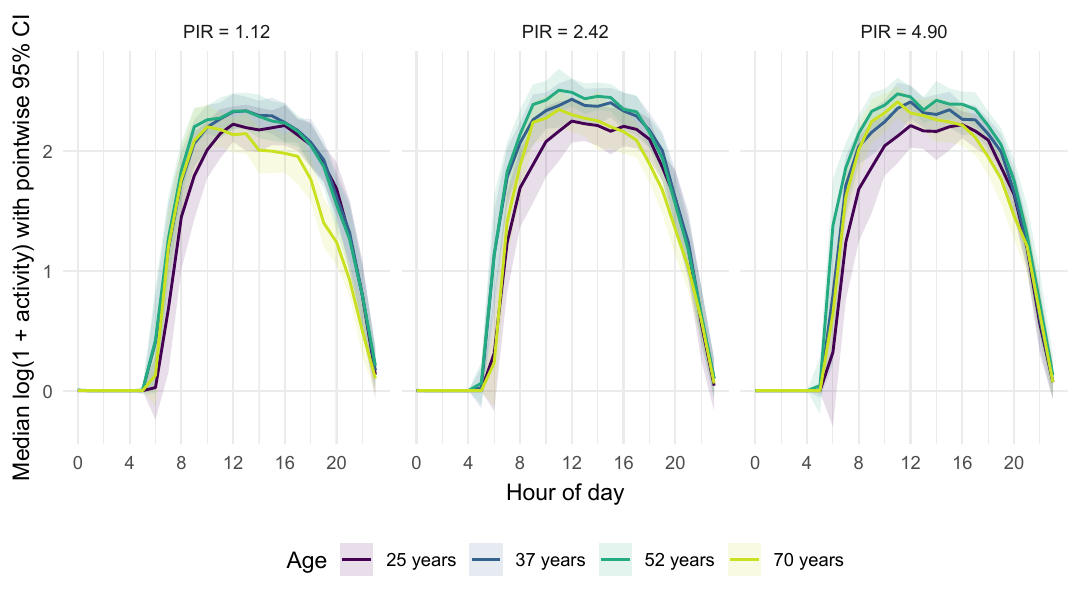}
  \caption{Pointwise median day-averaged log activity trajectories with pointwise 95\% confidence bands. Facets show poverty-income ratio and color encodes age.}
  \label{fig:nhanes-raw-trajectory-median}
\end{figure}

\subsection{Codings for NHANES Covariates}

The conditional distribution models in the NHANES application use seven baseline covariates: smoking status, age, race/ethnicity, education, family poverty-income ratio, marital status, and sex. These are the variables \texttt{smoke}, \texttt{age}, \texttt{race}, \texttt{edu}, \texttt{income}, \texttt{marital}, and \texttt{sex} in the analysis code provided in the Supplementary Material. Codings for these variables are given below.

\begin{itemize}
\item \texttt{smoke}: Smoking status is derived from \texttt{SMQ040}, ``Do you now smoke cigarettes?'' We define \texttt{smoke}=1 for respondents who report smoking every day and \texttt{smoke}=0 otherwise.
\item \texttt{age}: Age is measured in years with values 85 and above top-coded at 85. 
\item \texttt{race}: Race takes values in $\{1, 2, 3, 4, 5\}$ with 1 = Mexican American, 2 = Other Hispanic, 3 = Non-Hispanic White, 4 = Non-Hispanic Black, and 5 = Other race, including multi-racial.
\item \texttt{edu}: Adult education takes values in $\{1, 2, 3, 4, 5\}$ with 1 = less than 9th grade, 2 = 9--11th grade including 12th grade with no diploma, 3 = high school graduate/GED or equivalent, 4 = some college or AA degree, 5 = college graduate.
\item \texttt{income}: Family income is represented by the family poverty-income ratio. This ranges from 0 to 5, with values 5 and above top-coded at 5. 
\item \texttt{marital}: Marital status takes values in $\{1, 2, 3, 4, 5, 6\}$ with 1 = married, 2 = widowed, 3 = divorced, 4 = separated, 5 = never married, and 6 = living with partner.
\item \texttt{sex}: Sex takes values in $\{1, 2\}$ with 1 = male and 2 = female.
\end{itemize}

There was limited missingness in the variables used in our analysis, and a complete-case analysis was performed.

\subsection{Coverage of Little-Bag Confidence Intervals}
\label{sec:simulation-coverage}

We now study whether the little-bag procedure described in Section~\ref{sec:inference-target-functionals} provides reliable uncertainty quantification for scalar functionals of the conditional law. 
We focus on targets of the form
\begin{math}
  \theta_g(x)
  =
  \mathbb E\{g(Y)\mid X=x\},
\end{math}
estimated by the fDRF plug-in estimator 
\begin{math}
  \widehat\theta_g(x)
  =
  \int g(y)\,\widehat F_x(dy)
  =
  \sum_{i=1}^N w_i(x)g(Y_i).
\end{math}
We generate covariates $X_i \sim \Uniform([0,1]^5)$ and the index
\begin{math}
  \eta_i
  =
  \sqrt{12}\,\beta^\top\left(X_i-\frac12\mathbf 1\right),
\end{math}
where
\begin{math}
  \beta
  =
  \frac{(1,0.8,-0.6,0.4,-0.2)^\top}
  {\|(1,0.8,-0.6,0.4,-0.2)\|_2},
\end{math}
We then evaluate coverage at the covariate values
\begin{math}
  x_\eta
  =
  \frac12\mathbf 1
  +
  \frac{\eta}{\sqrt{12}}\beta
\end{math}
for $\eta \in \{-1, 0, 1\}$. We consider sample sizes $N \in \{400, 800, 1600, 3200\}$, and for each replicated dataset we compute $\widehat \theta(x_\eta) \pm 1.96 \, \widehat{\operatorname{se}}\{\widehat\theta(x_\eta)\}$. In our main implementation we use $L_{\mathrm{bag}}=50$ little bags and $M_{\mathrm{bag}}=20$ trees per bag.

We consider three targets across three DGPs. The first is the exceedance functional $\theta(x) = \E\{\int_0^1 1\{Y_i(t) > 0.4\} \mid X_i = x\}$ under a heteroskedastic Gaussian-process model, so that the target depends on conditional dispersion rather than only on the conditional mean. The second is a derivative-based roughness functional $\theta(x) = \E[\log \{1 + \int_0^1 ( DY_i(t) )^2 \ dt\} \mid X_i = x]$ where $Dy(t)$ denotes the derivative of $y(t)$, under a Gaussian-process model in which the length-scale varies with $\eta$. The third is a peak-location probability
\begin{math}
  \theta(x) = \Pr\{\arg\max_t Y_i(t) < 0.5 \mid X_i = x\}
\end{math}
under a two-regime phase-mixture model. Because the argmax map is discontinuous, this third setting is interpreted as a nonregular stress test rather than as a setting in which the little-bag intervals should necessarily achieve nominal coverage.

Results are given in Figure~\ref{fig:coverage} for the \fdrf\ using the Sobolev norm with $\lambda = 0.5$ and the $L_2$ norm. We see good coverage of the interval estimates for both the exceedance and peak-location functionals for both kernels. The roughness functional is quite difficult to estimate for both kernels, although we do see the weighted Sobolev kernel is able to cover the functional at all values of $\eta$ when the sample size is large. By contrast, the coverage for the $L_2$ kernel on this function is poor; this is because the $L_2$ kernel is not as sensitive as the Sobolev kernel to changes in the roughness of functions. This provides evidence that the choice of kernel is important for obtaining good coverage for estimators in addition to leading to better performance as measured by proper scoring rules.

\begin{figure}
  \centering
  \includegraphics[width=1\textwidth]{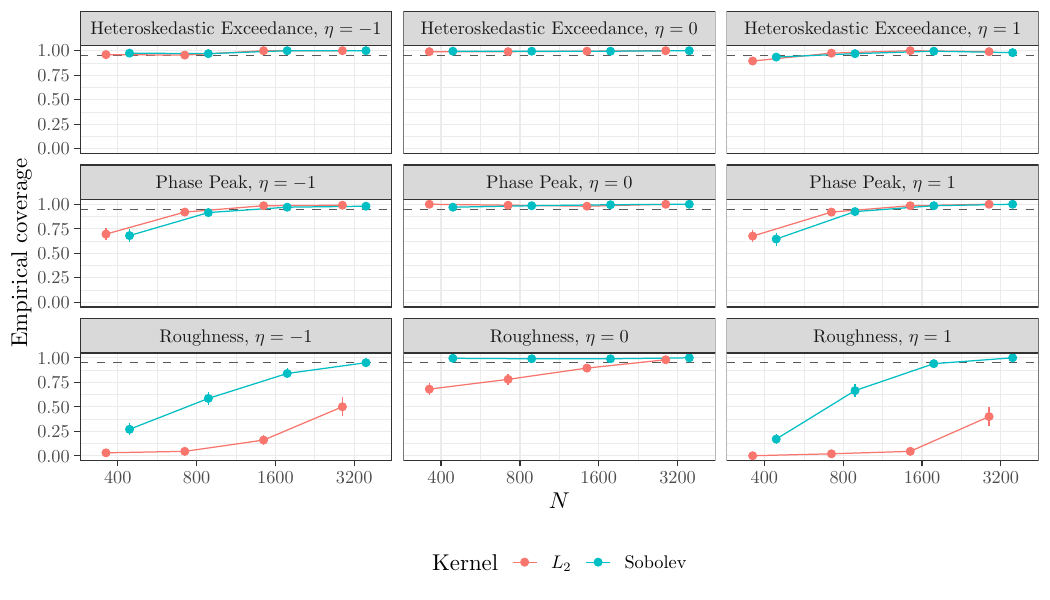}
  \caption{Empirical coverage of nominal 95\% confidence intervals under the settings described in Section~\ref{sec:simulation-coverage}.\label{fig:coverage}}
\end{figure}

Full details of the data-generating mechanisms are given in Supplementary Section~\ref{sec:supp-coverage-details}.

\section{Additional Details for the Coverage Simulation}
\label{sec:supp-coverage-details}

This section gives the full details for the coverage experiment summarized in Section~\ref{sec:simulation-coverage}. The goal is to assess whether the little-bag procedure used for fDRF functionals provides calibrated frequentist confidence intervals for scalar features of the conditional law. Throughout, the estimand is a population-level scalar functional. We focus on targets of the form
\begin{math}
    \theta_g(x)
    =
    \mathbb E\{g(Y)\mid X=x\},
\end{math}
for which the fDRF plug-in estimator is
\begin{math}
    \widehat\theta_g(x)
    =
    \int g(y)\,\widehat F_x(dy)
    =
    \sum_{i=1}^N w_i(x)g(Y_i).
\end{math}
After applying the scalar map $g(\cdot)$, the estimator is an honest forest weighted average of scalar pseudo-outcomes $g(Y_i)$.

\subsection{Common simulation setup}
In all settings, covariates are generated as
\begin{math}
    X_i \sim \Uniform([0,1]^5),
\end{math}
and the conditional law of the response depends on $X_i$ through the index
\[
    \eta_i
    =
    \sqrt{12}\,\beta^\top\left(X_i-\frac12\mathbf 1\right),
    \qquad
    \beta
    =
    \frac{(1,0.8,-0.6,0.4,-0.2)^\top}
    {\|(1,0.8,-0.6,0.4,-0.2)\|_2}.
\]
Coverage is evaluated at the three fixed covariate values
\[
    x_\eta
    =
    \frac12\mathbf 1
    +
    \frac{\eta}{\sqrt{12}}\beta,
    \qquad
    \eta\in\{-1,0,1\}.
\]
These values are interior to the covariate space and correspond to low, central, and high values of the signal index. All examples use the approximate MMD split criterion with $B = 200$ random features and consider 200 replicated datasets.

\subsection{Data Generating Process 1: Heteroskedastic Exceedance}

The first data-generating mechanism is designed to test coverage for a bounded functional that depends on conditional dispersion. We generate
\begin{math}
  Y_i(t)
  =
  \mu(t)
  +
  \sigma(\eta)Z_i(t),
\end{math}
where $Z_i$ is a zero-mean squared-exponential Gaussian process with unit marginal variance and length-scale $0.2$. That is,
\begin{math}
  \Cov\{Z_i(t),Z_i(t')\}
  =
  \exp\left\{
    -\frac{(t-t')^2}{2(0.2)^2}
  \right\}.
\end{math}
The mean function is
\begin{math}
    \mu(t)
    =
    0.5\sin(2\pi t)
    +
    0.3\sin(4\pi t),
\end{math}
and the conditional scale is
\begin{math}
    \sigma(\eta)
    =
    0.25
    +
    \frac{0.55}{1+\exp(-1.5\eta)}.
\end{math}
The target functional is the fraction of the domain over which the curve exceeds the
threshold $0.4$,
\[
  \theta(x) =
  \E\left\{
    \int_0^1 \mathbf 1\{Y_i(t)>0.4\}\,dt
  \mid X_i = x\right\}
  =
  \int_0^1 \Phi\left\{ \frac{\mu(t)-0.4}{\sigma(\eta(x))} \right\}
\]
This setting is intended to be favorable for the little-bag procedure: the target is bounded and smooth as a function of the index $\eta$, but it is still distributional because changing $\sigma(\eta)$ changes exceedance probabilities without changing the conditional mean function.

\subsection{Data Generating Process 2: Roughness Functional}

The second data-generating mechanism is designed to test uncertainty quantification for a
derivative-based feature of the conditional law. We generate
\begin{math}
    Y_i(t)
    =
    0.4\sin(2\pi t)
    +
    0.4Z_{i, \ell(\eta)}(t),
\end{math}
where $Z_{i, \ell}$ is a zero-mean squared-exponential Gaussian process with unit marginal
variance and length-scale
\begin{math}
    \ell(\eta)
    =
    0.05
    +
    \frac{0.25}{1+\exp(1.5\eta)}.
\end{math}
Thus larger values of $\eta$ correspond to shorter length-scales and rougher sample paths. The target functional is
\begin{align*}
    \theta(x)
    =
    \mathbb E\{g(Y_i)\mid X_i = x\}
  \quad \text{where} \quad
    g(y)
    =
    \log\left\{
        1+\int_0^1 [Dy(t)]^2\,dt
    \right\},
\end{align*}
where $Dy(t)$ is the derivative of $y(t)$. The true value of $\theta(x)$ is computed by Monte Carlo on a large dataset. This setting is especially relevant for functional distributional estimation because the signal is in the covariance and smoothness of the response, not in its conditional mean.

The purpose of this comparison is to distinguish failures of interval calibration from failures due to bias in the estimated conditional law. If the $L_2$-kernel forest does not split effectively on roughness differences, it may produce biased point estimates and undercoverage even if the little-bag standard error accurately estimates the variability around the biased estimator.

\subsection{Data Generating Mechanism 3: Phase-Mixture Peak Probability}

The third data-generating mechanism is a nonregular stress test motivated by our peak-time summaries. We generate a latent mixture indicator
\begin{math}
  Z_i\sim\operatorname{Bernoulli}\{p(\eta_i)\},
  \
  p(\eta)
  =
  \frac{1}{1+\exp(-1.5\eta)}.
\end{math}
Conditional on $Z_i = 1$, the curve has a peak centered near $0.35$, while conditional on $Z_i = 0$, the curve has a peak centered near $0.65$. Specifically,
\[
    Y_i(t)
    =
    A_i
    \exp\left\{
        -\frac{(t-C_i)^2}{2W_i^2}
    \right\}
    +
    0.03 \, \varepsilon_i(t),
\]
where $\varepsilon_i$ is a mean-zero Gaussian process with covariance $\kappa(t,t') = \exp\left\{\frac{-(t - t')^2 }{2 \ell^2}\right\}$ and $\ell = 0.12$. The amplitude and width are generated as
\[
    A_i=\exp(0.10\xi_i),
    \qquad
    W_i=0.06\exp(0.10\zeta_i),
    \qquad
    \xi_i,\zeta_i\sim \Normal(0,1).
\]
The center is set to $C_i = 0.65 - 0.3 Z_i + \delta_i$ where $\delta_i \sim \Normal(0, 0.025^2)$.

The target functional is the probability that the curve peaks in the first half of the domain,
\[
    \theta(x)
    =
    \Pr\{\arg\max_t Y_i(t)<0.5\mid X_i=x\}.
\]
Ties are broken by taking the first maximizer on the grid. The truth \(\theta(x_\eta)\) is computed by Monte Carlo on a large synthetic dataset. Unlike the preceding two targets, \(g_{\mathrm{peak}}\) is discontinuous as a functional of the curve: small perturbations can change the argmax when there are nearly tied peaks. We therefore interpret this setting as a stress test for the little-bag intervals rather than as a setting in which nominal coverage is guaranteed by standard smooth-functional heuristics.

\section{Theoretical Arguments for Confidence Intervals}

This appendix gives formal results for asymptotic normality of estimates, which provide some justification for our confidence interval construction. For the purpose of this section, we will analyze the ideal infinite-tree variant of the \fdrf\ so that we do not need to include a correction for within-bag Monte Carlo variability. To do this, we reduce each specified low-dimensional estimand to an honest regression forest with an appropriate scalar pseudo-response.

Let $g: \sY \to \Reals^m$ and let
\begin{math}
  \eta_g(x)
  =
  \E\{g(Y)\mid X=x\},
\end{math}
and $\widehat\eta_g(x) = \sum_{i=1}^N w_i(x)g(Y_i)$. For $a \in \Reals^m$, let $Z_a = a^\top g(Y)$, $m_a(x) = \E(Z_a \mid X = x)$, and $m_{2,a}(x) = \E(Z_a^2 \mid X = x)$. For a scalar target $\theta(x) = \Psi\{\eta_g(x)\}$ with estimator $\widehat \theta(x) = \Psi\{\widehat \eta_g(x)\}$, let $a_x = \left.\nabla \Psi(\eta)\right|_{\eta = \eta_g(x)}$. We make the following additional assumptions.

\begin{assumption}[Targetwise forest regularity]
  \label{assumption:7-inference}
  Fix $x \in [0,1]^P$ in the interior and suppose Assumption~\ref{assumption:1-rf} and Assumption~\ref{assumption:3-covariate} hold. Additionally, the trees use subsamples of size $s_N \asymp N^\beta$ where
  \[
    \beta_{\min}
    =
    \left[
      1+
      \frac{\pi}{P}
      \frac{
        \log\{(1-\alpha)^{-1}\}
      }{
        \log(\alpha^{-1})
      }
    \right]^{-1}
    <
    \beta
    <
    1.
  \]
\end{assumption}

\begin{assumption}[Properties of the Outcome Distribution]
  \label{assumption:8-inference}
  For every direction
  \begin{math}
    a\in\mathcal A_x
    =
    \{e_1,\ldots,e_m,a_x\},
  \end{math}
  where $e_j$ denotes the $j$th coordinate vector, the functions $m_a$ and $m_{2,a}$ are Lipschitz on $[0,1]^P$ and, for some $\delta>0$ and $C<\infty$,
  \[
    \sup_{x\in[0,1]^P}
    \E\left[
      |Z_a-m_a(x)|^{2+\delta}
      \mid X=x
    \right]
    \le C.
  \]
  Finally, for all $a \in \mathcal A_x$, $\Var(Z_a \mid X = x) > 0$.
\end{assumption}

\begin{assumption}[Assumptions on $\Psi$]
  \label{assumption:9-inference}
  The gradient $\nabla \Psi$ is Lipschitz in a neighborhood of $\eta_g(x)$.
\end{assumption}

\begin{proposition}[Scalar pseudo-response reduction]
\label{prop:scalar-pseudo-response}
Under the Assumption~\ref{assumption:7-inference}---\ref{assumption:9-inference}, for every fixed $a\in\mathcal A_x$ there is a sequence $\sigma_{N,a}(x)\to0$ such that
\begin{equation}
  \label{eq:scalar-pseudo-clt}
  \frac{
    a^\top\{\widehat\eta_g(x)-\eta_g(x)\}
  }{
    \sigma_{N,a}(x)
  }
  \cind
  \Normal(0,1).
\end{equation}
Here $\sigma_{N,a}^2(x)$ is the asymptotic sampling variance of the
honest regression forest whose scalar pseudo-response is
$Z_a=a^\top g(Y)$.

\end{proposition}

\begin{proof}
  Assumptions~\ref{assumption:7-inference}---\ref{assumption:9-inference} are sufficient to satisfy the hypothesis of Theorem~3.1 of \citet{wager2018estimation}. Applying that theorem to $Z_a$ proves \eqref{eq:scalar-pseudo-clt}.
\end{proof}

\begin{proposition}[Smooth functions of conditional moments]
\label{prop:smooth-moment-target}
Under Assumptions~\ref{assumption:7-inference}---\ref{assumption:9-inference},
\begin{equation}
  \label{eq:smooth-target-clt}
  \frac{
    \widehat\theta(x)-\theta(x)
  }{
    \sigma_{N,a_x}(x)
  }
  \cind
  \Normal(0,1).
\end{equation}
\end{proposition}

\begin{proof}
Let
\begin{math}
  \Delta_N(x)
  =
  \widehat\eta_g(x)-\eta_g(x).
\end{math}
Because $\nabla\Psi$ is locally Lipschitz, a Taylor expansion with
integral remainder gives, on an event whose probability tends to one,
\[
  \left|
    \Psi\{\eta_g(x)+\Delta_N(x)\}
    -
    \Psi\{\eta_g(x)\}
    -
    a_x^\top\Delta_N(x)
  \right|
  \le
  C_\Psi\|\Delta_N(x)\|^2
\]
for a finite constant $C_\Psi$. Lemma~\ref{lem:forest-rates} below implies
that the right-hand side is
$o_p\{\sigma_{N,a_x}(x)\}$. Hence
\[
  \widehat\theta(x)-\theta(x)
  =
  a_x^\top\Delta_N(x)
  +
  o_p\{\sigma_{N,a_x}(x)\}.
\]
Proposition~\ref{prop:scalar-pseudo-response}, applied with
$a=a_x$, and Slutsky's theorem yield
\eqref{eq:smooth-target-clt}.
\end{proof}

\begin{remark}
  Proposition~\ref{prop:smooth-moment-target} does not follow immediately from the delta method because Proposition~\ref{prop:scalar-pseudo-response} does not establish joint asymptotic normality of $\widehat \eta_g(x)$ under a common scaling.
\end{remark}

\begin{lemma}[Rates needed for the nonlinear remainder]
\label{lem:forest-rates}
Let
\begin{math}
  r_N=\sqrt{\frac{s_N}{N}}.
\end{math}
Under Assumptions~\ref{assumption:7-inference}---\ref{assumption:9-inference},
\[
  \|\widehat\eta_g(x)-\eta_g(x)\|
  =
  O_p(r_N)
  \quad \text{and} \quad
  \sigma_{N,a_x}(x)
  \gtrsim
  \frac{
    r_N
  }{
    \{\log(s_N)\}^{P/2}
  }.
\]
Consequently, $\|\widehat \eta_g(x) - \eta_g(x)\|^2 = o_P\{\sigma_{N,a_x}(x)\}$.
\end{lemma}

\begin{proof}
Write
\begin{math}
  \widehat m_{N,a}(x)
  =
  a^\top\widehat\eta_g(x)
  =
  \sum_{i=1}^N w_i(x)\,a^\top g(Y_i)
\end{math}
for the forest prediction based on the scalar pseudo-response
\begin{math}
  Z_a=a^\top g(Y).
\end{math}
Let $T_{s_N,a}(x)$ denote the corresponding prediction from one honest tree grown on a subsample of size $s_N$.

First consider a coordinate direction $a=e_j$. The infinite
subsampled forest is a complete $U$-statistic of order $s_N$.
The standard Hoeffding variance decomposition \citep[][Theorem~5.2]{hoeffding1948class} therefore gives
\begin{equation}
  \label{eq:coordinate-forest-variance}
  \Var\!\left\{
    \widehat m_{N,e_j}(x)
  \right\}
  \le
  \frac{s_N}{N}
  \Var\!\left\{
    T_{s_N,e_j}(x)
  \right\}.
\end{equation}
By Assumption~\ref{assumption:8-inference} we have
\begin{math}
  \sup_x
  \E\{g_j(Y_i)^2\mid X_i=x\}<\infty,
\end{math}
so that the single-tree second moment and its variance are bounded uniformly in $N$ (recall that by Assumption~\ref{assumption:1-rf} the leaf nodes in a single tree contain between $\kappa$ and $2\kappa - 1$ observations). It follows from \eqref{eq:coordinate-forest-variance} that
\[
  \widehat m_{N,e_j}(x)
  -
  \E\{\widehat m_{N,e_j}(x)\}
  =
  O_p\left(
    \sqrt{\frac{s_N}{N}}
  \right).
\]
To control the bias, define
\begin{math}
  \lambda_\alpha
  =
  \frac{\pi}{P}
  \frac{
    \log\{(1-\alpha)^{-1}\}
  }{
    \log(\alpha^{-1})
  }.
\end{math}
Applied to the scalar response $g_j(Y)$,
Theorem~3.2 of \citet{wager2018estimation} gives
\[
  \left|
    \E\{\widehat m_{N,e_j}(x)\}
    -
    \eta_{g,j}(x)
  \right|
  =
  O\left(
    s_N^{-\lambda_\alpha/2}
  \right).
\]
Because $s_N\asymp N^\beta$ and
\begin{math}
  \beta>
  \frac{1}{1+\lambda_\alpha},
\end{math}
we have
\begin{math}
  s_N^{-\lambda_\alpha/2}
  =
  o\!\left(
    \sqrt{\frac{s_N}{N}}
  \right).
\end{math}
Thus, with
\begin{math}
  r_N=\sqrt{\frac{s_N}{N}},
\end{math}
we have
\begin{math}
  \widehat\eta_{g,j}(x)-\eta_{g,j}(x)
  =
  O_p(r_N)
\end{math}
for every $j=1,\ldots,m$. Hence 
\begin{equation}
  \label{eq:moment-vector-upper-rate}
  \|
    \widehat\eta_g(x)-\eta_g(x)
  \|
  =
  O_p(r_N).
\end{equation}
It remains to obtain a lower bound for the standard deviation in the linearized direction
\begin{math}
  a_x=\nabla\Psi\{\eta_g(x)\}.
\end{math}
Let $\mathring T_{s_N,a_x}(x)$ denote the centered first-order
Hájek projection of the single tree
$T_{s_N,a_x}(x)$. The variance of the leading Hájek term of the
subsampled forest is
\begin{equation}
  \label{eq:forest-hajek-variance}
  \sigma_{N,a_x}^2(x)
  =
  \frac{s_N}{N}
  \Var\{
    \mathring T_{s_N,a_x}(x)
  \}
  =
  \frac{s_N^2}{N}
  \Var\left[
    \E\{
      T_{s_N,a_x}(x)
      \mid Z_1
    \}
  \right].
\end{equation}
We note that, because $s_N(\log N)^{P}/N\to 0$, \citet[Lemma~3.3 and Theorem~3.4] {wager2018estimation} implies that $\sigma_{N,a_x}$ is a valid normalization for the scalar forest central limit theorem. 

By \citet[Corollary~3]{wager2018estimation}, a double-sample,
honest, regular, symmetric random-split tree is
$\nu(s_N)$-incremental, i.e.,
\[
  \frac{
    \Var\{
      \mathring T_{s_N,a_x}(x)
    \}
  }{
    \Var\{
      T_{s_N,a_x}(x)
    \}
  }
  \gtrsim
  \frac{1}{
    \{\log(s_N)\}^{P}
  }.
\]
Moreover, the assumption
\begin{math}
  \Var\{a_x^\top g(Y)\mid X=x\}>0
\end{math}
implies that the single-tree variance is bounded away from zero.
Indeed, honesty makes the estimation responses conditionally
independent of the selected leaf; the estimation leaf contains at
most $2\kappa-1$ observations, and the leaf containing $x$ shrinks to
$x$ with high probabiliy \citep[][Lemma~1]{wager2018estimation}.
Therefore
\begin{math}
  \liminf_{N\to\infty}
  \Var\{
  T_{s_N,a_x}(x)
  \}
  >0.
\end{math}
Combining this fact with
\eqref{eq:forest-hajek-variance} and the incrementality bound yields
\begin{equation}
  \label{eq:linearized-standard-deviation-lower}
  \sigma_{N,a_x}(x)
  \gtrsim
  \frac{
    \sqrt{s_N/N}
  }{
    \{\log(s_N)\}^{P/2}
  }.
\end{equation}
Finally, \eqref{eq:moment-vector-upper-rate} and
\eqref{eq:linearized-standard-deviation-lower} give
\[
  \frac{
    \|
      \widehat\eta_g(x)-\eta_g(x)
    \|^2
  }{
    \sigma_{N,a_x}(x)
  }
  =
  O_p\!\left[
    \sqrt{\frac{s_N}{N}}\,
    \{\log(s_N)\}^{P/2}
  \right]
  =
  o_p(1),
\]
because $s_N\asymp N^\beta$ with $\beta<1$ and $P$ is fixed.
\end{proof}

\begin{remark}
  The presentation here focuses on (functions of) linear functionals as inferential targets. Under suitable conditions, this covers means, probabilities, variances, and circular summaries. Our NHANES analysis also considers nonlinear functionals, such as conditional quantiles. It is possible to handle sufficiently regular nonlinear functionals that are solutions to estimating equations by performing a suitable asymptotic expansion; this is also done by \citet{athey2019generalized}. We defer details to that source.
\end{remark}


\begin{remark}
  We have also elected to provide details for the case of a single design point $x_i$. The NHANES application also considers \emph{contrasts} between different design points, such as the difference in exceedance probabilities between two covariate values. The results above can be extended cover such contrasts using results for generalized $U$-statistics given by \citet{peng2022rates}, as (after linearization) the difference of two honest forest-based estimands is itself a generalized $U$-statistic. In particular, Theorem~1 of \citet{peng2022rates} gives the corresponding asymptotic distribution whenever the difference is $\nu(s_N)$-incremental in the sense of \citet{wager2018estimation} and
  \[
      \frac{s_N}{N\nu(s_N)}\longrightarrow 0.
  \]
  The incrementality result for regression trees of $\nu(s_N)\gtrsim(\log s_N)^{-P}$ established by \citet{wager2018estimation} reduces this requirement to $s_N(\log s_N)^P/N \to 0$; this result is not immediate from \citet{wager2018estimation}, however, as it requires verifying the incrementality condition for the \emph{difference} of two forest evaluations. We do not pursue this extension formally here.

\end{remark}

\subsection{Verification of Assumptions for Specific Examples}

\paragraph{Conditional probabilities.}
For an event $A$ determined by an individual curve, take
\begin{math}
  g_A(y)=1(y\in A).
\end{math}
Then
\begin{math}
  \widehat p_A(x) = \sum_{i = 1}^N w_i(x) 1(Y_i \in A)
\end{math}
is an honest regression-forest prediction. The required conditions reduce to Lipschitz continuity of
\begin{math}
  x\mapsto\Pr(Y\in A\mid X=x)
\end{math}
and nondegeneracy
\begin{math}
  0<\Pr(Y\in A\mid X=x)<1.
\end{math}

\paragraph{Mean and Standard Deviation of Maximal Effort}
Let $M(y) = \max_t y(t)$ and $g_M(y) = (M(y), M(y)^2)$. Write $c(x) = \E\{M(Y_i) \mid X_i = x\}$ and $d(x) = \E\{M(Y_i)^2 \mid X_i = x\}$. A sufficient moment condition is
\[
  \sup_{z\in[0,1]^P}
  \E\{|M(Y_i)|^{4+2\delta}\mid X_i = x\}
  <
  \infty
\]
for some $\delta>0$. We also assume that the conditional moments
\begin{math}
  x \mapsto \E\{M(Y_i)^r\mid X_i=x\},
\end{math}
are locally Lipschitz for $r = 1,\ldots,4$. These conditions imply the required first and second moment regularity for every scalar linear combination of $M$ and $M^2$ (hence, the mean and standard deviation).

For the conditional mean, the linearized pseudo-response is $M$,
and nondegeneracy requires
\begin{math}
  \Var(M(Y_i)\mid X_i=x)>0.
\end{math}
For the conditional standard deviation, the linearized pseudo-response is
\[
  -\frac{c(x)}{\sigma_M(x)}M
  +
  \frac{1}{2\sigma_M(x)}M^2,
\]
whose variance is, up to a nonzero multiplicative constant, the variance of $\{M-c(x)\}^2$. Thus a nondegenerate first-order interval for the conditional standard deviation requires
\begin{math}
  \Var\left(
    \{M(Y_i) - c(x)\}^2
    \mid X_i = x
  \right) > 0.
\end{math}

\paragraph{Peak time.}
Let $\mathcal T\subset[0,24)$ denote the observation grid and set $P(Y_i)=\min\{t\in\mathcal T:Y_i(t)=\max_{s\in\mathcal T}Y_i(s)\}$, so that ties are resolved by selecting the earliest maximizing grid time, and set
\begin{math}
  g_P(Y_i) = \binom{\cos\{\pi P(Y_i) / 12\}}{\sin\{\pi P(Y_i) / 12\}}.
\end{math}
Write
Write
\[
  \eta_P(x)
  =
  \E\{g_P(Y_i) \mid X_i = x\}
  =
  \begin{pmatrix}
    a(x)\\ b(x)
  \end{pmatrix},
  \qquad
  r(x)
  =
  \{a(x)^2+b(x)^2\}^{1/2}.
\]
The circular mean and standard deviation are
\[
  \Psi_P\{a(x),b(x)\}
  =
  \left[
    \left[\frac{12}{\pi}\operatorname{atan2}\{b(x),a(x)\}\right]_{24}, \quad
    \frac{12}{\pi}\sqrt{-2\log r(x)}
  \right],
\]
where $[u]_{24}$ return $u$ mod $24$. The transformation $g_P$ is bounded, and each coordinate of $\Psi_P$ is continuously differentiable whenever $0<r(x)<1$, after choosing a local angular chart for the circular mean. Hence, the required conditions for inference hold.

\end{document}